\documentclass[a4paper,reqno]{amsart}
\usepackage[english]{babel}
\usepackage[T1]{fontenc}
\usepackage{amsmath,amssymb,amsfonts,amsthm}
\usepackage{hyperref}
\usepackage{slashed}
\usepackage{graphicx,color}
\usepackage[bbgreekl]{mathbbol}
\usepackage{mathtools}

\numberwithin{equation}{section}

\newtheorem{definition}{Definition}[section]
\newtheorem{theorem}[definition]{Theorem}
\newtheorem{proposition}[definition]{Proposition}
\newtheorem{lemma}[definition]{Lemma}
\newtheorem{corollary}[definition]{Corollary}

\newtheorem{remark}[definition]{Remark}

\newcommand{\X}[1]{(X_{#1})}

\newcommand{\qsp}[2]{\,\ensuremath{\raise.5ex\hbox{$#1$}\big\slash\raise-.5ex\hbox{$#2$}}} 

\newcommand{\pard}[2]{\frac{\delta#1}{\delta#2}}

\newcommand{\intl}{\int\limits}

\newcommand{\tom}{\widetilde{\omega}}
\newcommand{\bom}{\boldsymbol{\omega}}

\newcommand{\bem}{\mathbf{e}}
\newcommand{\tem}{\underline{\mathbf{e}}}
\newcommand{\binv}{[\underline{\mathbf{e}}^{-1}]}
\newcommand{\bA}{\mathbf{A}}

\newcommand{\bt}{\mathbf{t}}
\newcommand{\bg}{\mathbf{\Gamma}}

\newcommand{\Wedge}[1]{{\textstyle \bigwedge^{#1}}}

\title[Classical Palatini--Cartan--Holst]{The reduced phase space of Palatini--Cartan--Holst theory}
\author{A. S. Cattaneo}
\address{Institut f\"ur Mathematik, Winterthurerstrasse 190, 8057 Z\"urich, Switzerland}
\email{cattaneo@math.uzh.ch}

\author{M. Schiavina}
\address{Department of Mathematics, University of California, Berkeley, 970 Evans Hall, Berkeley 94720, U.S.A.}
\email{michele.schiavina@berkeley.edu}
\date{\today}

\thanks{This research was (partly) supported by the NCCR SwissMAP, funded by the Swiss National Science Foundation and by the COST Action MP1405 QSPACE, supported by COST (European Cooperation in Science and Technology).
A. S. C. acknowledges partial support of SNF grant No. 200020\_172498$\slash$1. M. S. is supported by SNF grant No. P2ZHP2\_164999. }

\begin{document}

\begin{abstract}
General relativity in four dimensions can be reformulated as a gauge theory, referred to as Palatini--Cartan--Holst theory. This paper describes its reduced phase space using a geometric method due to Kijowski and Tulczyjew and its relation to that of the Einstein--Hilbert approach. 
\end{abstract}

\maketitle

\tableofcontents

\section{Introduction}
General Relativity (GR) is defined in terms of the metric tensor and of the Einstein--Hilbert (EH) action functional. A (classically) alternative way of formulating it, which has the advantage of being a gauge theory, follows from the observation that the dynamical metric may be expressed in terms of a fixed reference metric via a dynamical (co)frame field \cite{Cartan}.  We will refer to this version of GR as Palatini--Cartan--Holst (PCH) theory as detailed below, in a discussion about nomenclature.

The reduced phase space of a theory is the space of possible initial conditions endowed with its natural symplectic structure\footnote{Traditionally, the reduced phase space is defined as the space of solutions endowed with its natural symplectic structure. If the theory is formulated on a manifold of the form $\Sigma\times[a,b]$ and $\Sigma$ is a Cauchy surface, this is the same as the space of possible initial conditions on $\Sigma$. We use a more general definition where $\Sigma$ is not necessarily Cauchy. In particular by initial conditions we mean conditions for which there is a, possibly non unique, local, but not necessarily global, solution.}. For example, in the usual case of mechanics on a target manifold $M$, it is $T^*M$ with its canonical symplectic structure. In the case of electromagnetism in four dimensions, one starts with a phase space in which the conjugate variables are the vector potential and the electric field on the initial $3$-surface with symplectic form induced from their pairing. The reduced phase space is then given by the solution to the Gauss law (vanishing of the divergence of the electric field) modulo gauge transformations. In the case of GR, in space--time dimension greater than two, one starts with a phase space presented as the cotangent bundle of the space of metrics on the initial space-like hypersurface. The reduced phase space is then obtained as the solutions to the 
so-called energy and momentum constraints modulo diffeomorphisms, both tangential to the space hypersurface and transversal to it.

One well-known method to obtain the reduced phase space, which works well in the above examples, is due to Dirac \cite{Dirac}. The literature is also  full of attempts to apply this method to  PCH theory.

In this paper, we will study the reduced phase space of PCH theory using instead a geometric method introduced by Kijowski and Tulczyjew \cite{KT} (which also has the advantage of being compatible with the BV-BFV formalism introduced in \cite{CMR1,CMRCorfu}). We will show that, under the assumption that the induced boundary metric is non degenerate, the reduced phase space can be nicely described and corresponds indeed to that of the EH formulation, which has two local degrees of freedom (interpreted as the two possible polarizations of the graviton). Note that this assumption is just an open condition on the space of bulk co-frame fields. We do not compute the reduced phase space without this assumption, but a result proved below suggests that in the case of a light-like boundary the reduction should have no local degrees of freedom.

In a nutshell\footnote{In this Introduction, for simplicity we do not present the extension by Holst depending also on the inner dual of the curvature, which is discussed in details in the rest of the paper.} our result is as follows. We start with the action functional
\[
S[e,\omega]=\int\limits_M \mathrm{Tr}\left[ e\wedge e \wedge F_\omega + \frac{\Lambda}{4} e^4\right],
\]
where $M$ is a four-dimensional manifold with boundary (which admits Lorentzian structures) endowed with a rank-four vector bundle isomorphic to $TM$ with a reference fibre metric, $e$ a tetrad, $\omega$ an orthogonal connection and $\Lambda$ the cosmological constant. We assume that also the boundary restriction of the metric induced by $e$ is nondegenerate. Then our result is that the reduced phase space is obtained by coisotropic reduction in the symplectic space consisting of the  space of boundary tetrads, denoted by $\mathbf{e}$, and of boundary connections modulo the action of $e\wedge\cdot$ (Theorem \ref{Theo:PCHClassical}). Denote by $\bom=[\omega]_e$ the respective equivalence class, the symplectic structure $\omega^\partial=\delta\alpha^\partial$ reads
\[
\alpha^{\partial}=\frac12\int\limits_{\partial M}\mathrm{Tr}\left[\mathbf{e}\wedge\mathbf{e}\wedge\delta\bom \right].
\]

We show that this reduction is equivalent to the space of boundary tetrads and connections $(e,\omega)$ satisfying the structural constraint $p d_{\omega} e = 0$, where $p$ is the projection (relying on an irrelevant choice of complement) to the space of forms $\beta$ satisfying $e \wedge \beta = 0$, and using this description we are able to prove that the constraints defining the coisotropic submanifold are
\[
\mathbf{L}_\alpha=\intl_{\partial M} \mathrm{Tr}[\alpha\, \mathbf{e}\wedge d_{\bom}\mathbf{e}],\qquad 
\mathbf{J}_\mu = \intl_{\partial M}\mathrm{Tr}\left[ \mu \left(\bem\wedge F_{\bom} + \Lambda\bem^3\right)\right],
\]
where $\alpha$ and $\mu$ are Lagrange multipliers (Theorem \ref{constraintstheorem} and Corollary \ref{constraintscorollary}). The reader should not be misled by the apparent simplicity of the constraints that seem to look like the restriction to the boundary of the Euler--Lagrange equations: the crucial point here is that the reduced connection appears. One helpful way of looking at this is to consider the equation $d_{\omega}e =0 $ coming from the bulk Euler--Lagrange equations and split it into the structural constraint $pd_{\omega} e=0$, and the residual constraint $\mathbf{L}_\alpha$. The nontrivial part of the proof consists in showing this result (which moreover holds only under the assumption that the boundary metric is nondegenerate). 

The constraints $\mathbf{L}_\alpha$ are an equivariant momentum map for the action of internal gauge transformation. The corresponding Marsden--Weinstein \cite{MW} reduction yields the space of boundary fields for EH theory and the constraints $\mathbf{J}_\mu$ descend there to the usual energy and momentum constraints (Theorem \ref{PCHtoEH}). A different partial reduction using only half of the  $\mathbf{L}_\alpha$ constraints leads to Ashtekar's formulation in terms of tetrads and connections for the boundary orthogonal group.

To summarize, note that EH and PCH theories are equivalent \emph{on shell}, that is, on the critical locus of the action, where the equations of motion are solved, and up to the respective symmetries. Here we show that their boundary structures are also equivalent. The next step will be that of understanding their possible equivalence at the BV-BFV level, a project that we will start in
\cite{CS2} based on results from \cite{CS1}.

Finally, in Section \ref{clhalfshell}, we study a variant of the PCH action functional that enforces the compatibility between the connection and the co-frame field, fixing the former to be the Levi--Civita connection for $g=e^*\eta$ by means of a Lagrange multiplier. The classical phase space of the resulting theory will be a cotangent bundle, and it will be symplectomorphic to the space of boundary fields of PCH theory, yet the boundary structure will turn out to be different (and arguably ill-behaved), unless vanishing boundary conditions on the Lagrange multiplier are forced from the very beginning. As a matter of fact, the projection of the Euler--Lagrange equations to the boundary will be isotropic but not Lagrangian, and inequivalent to that of PCH. 
This shows that equivalent theories on closed manifolds without boundary may differ when boundaries are included, and will call for a refinement of the notion of classical equivalence.

\begin{remark}
In this note we focus on the physical case of Lorenzian signature. The Euclidean case follows automatically, with no restrictions on the boundary values of the co-frame fields, as the boundary metric is automatically nondegenerate.
\end{remark}

\begin{remark}
The co-frame field is required to be nondegenerate to establish the equivalence between the Euler--Lagrange loci modulo symmetries in the Palatini--Cartan and in the Einstein--Hilbert formulations. Unlike the three-dimensional case, where this nondegeneracy condition can be removed leading to a generalization of GR \cite{CSS}, we will see that in the four-dimensional case this condition is also necessary to make the boundary space of fields well-behaved and therefore must be maintained.
\end{remark}

\subsection*{Nomenclature}
In this work we refer to a field theory that carries the names of Cartan, Palatini and Holst. This theory describes General Relativity as a model for the gravitational interaction, a priori different from the original Einstein--Hilbert metric formulation.

The (controversial) origin of this name can be traced back to the (arguably) historically incorrect way that Palatini was attributed the idea of considering the connection as an independent field, as thoroughly discussed in \cite{FFR, BH}. In \cite{Palatini} it was observed that the variation of the Ricci tensor can be written in terms of the variation of the Christoffel symbols. However, the paper was interpreted as suggesting independence from the metric field, without ever suggesting it explicitly. 

Instead, the key observation came from Cartan and Einstein \cite{Cartan,Einstein}, and was later explained through the powerful language of moving frames and the tetrad formalism by several authors, most notably by Kibble and Sciama \cite{Kib, Sci} (see \cite{Wise} for a modern account on this topic).  This is usually referred to as Einstein--Cartan--Sciama--Kibble (ECSK) theory.

In principle, ECSK theory is more general than Einstein--Hilbert in that it admits nonvanishing torsion and coupling with spin matter but, as we are not concerned with coupling to matter in the present paper, we will not appreciate that feature. This has the advantage of deriving the corrected Bianchi identities in the presence of spin matter from the Euler--Lagrange equations.

However, we understand the ECSK theory as a \emph{metric} variational problem where the torsion field is considered to be a dynamical correction to the Levi--Civita connection. Although it can be phrased in terms of tetrad fields and principal connections, we wish to emphasise that formulations that are equivalent on-shell and up to possible boundary terms, need not be equivalent when such requirements are relaxed.

Moreover, it is common practice in the literature to refer to the formula that links the variation of the Ricci tensor to the variation of the Christoffel symbols as the \emph{Palatini identity}, and the formulation of GR that we will analyse here is also commonly referred to as \emph{Palatini gravity}. 

In this paper we adopt this convention in view of this (perhaps erroneous) tradition, but we credit Cartan at the same time. In this spirit we use the name Palatini--Cartan, which then should be understood as meaning \emph{tetrad formalism and independent principal $\mathfrak{so}(3,1)$ connection\footnote{Or $\mathfrak{so}(4)$ in the Euclidean version.}}.

Finally, the addition of the name of Holst will be justified in section \ref{Holstname}, where we will consider a generalisation of the Palatini--Cartan action, due to the fact that $\mathfrak{so}(3,1)$ admits two invariant inner products. As a matter of fact, the action functional corresponding to EH in the Palatini--Cartan formalism in four space-time dimensions that takes into account the Barbero--Immirzi parameter [Bar, Imm], is due Holst \cite{Holst} (see equation (3.1)), and originally observed in \cite{HMS} for metric gravity with torsion (Einstein--Cartan theory).

\subsection*{Acknowledgements} We thank Giovanni Canepa for several constructive discussions, and Friedrich Hehl for valuable comments about the controversy in the nomenclature. We thank G. Canepa and the anonymous referee for having found flaws in previous versions of the article.

\section{Geometric theory of boundary data}\label{Sect:Constraintstheory}

A classical field theory is the assignment of a space of fields $F_M$ and a local action functional $S_M$ to a manifold with boundary $(M,\partial M)$. Usually $F_M$ is taken to be some mapping space, or space of sections of a vector bundle or a sheaf. Integration by parts in the computation of the variation $\delta S_M$ defines a one-form $\widetilde{\alpha}_{\partial M}$ on the space of restrictions of fields (and jets) to the boundary $\widetilde{F}_{\partial M}$: i.e. $\delta S_M=\mathsf{el} + \widetilde{\pi}_{\partial M}^*\widetilde{\alpha}_{\partial M}$ (the vanishing locus of $\mathsf{el}$ is the Euler--Lagrange locus), and there is a surjective submersion $F_M \xrightarrow{\widetilde{\pi}_M}\widetilde{F}_{\partial M}$.

In this setting, we can also consider the space of fields associated with the collar $\partial M\times[0,\epsilon]$. Such space $F_{\partial M\times[0,\epsilon]}$ also maps to the space of pre-boundary fields $\widetilde{F}_{\partial M}$. Assuming that the kernel of the two form $\widetilde{\omega}_{\partial M}\coloneqq \delta\widetilde{\alpha}$ is regular (a subbundle of $T\widetilde{F}_{\partial M}$) and that the quotient by $\mathrm{ker}(\widetilde{\omega}_{\partial M})$ is smooth, we can construct the \emph{true} space of boundary fields:
\begin{equation}
F^\partial_{\partial M} \coloneqq \qsp{\widetilde{F}_{\partial M}}{\mathrm{ker}(\widetilde{\omega}_{\partial M})}.
\end{equation}

The critical locus of the action functional in the bulk is denoted by $EL_{M}$, and if we denote by $\pi_M\colon F_M\longrightarrow F^\partial_{\partial M}$ the canonical projection, we can construct $L_M\coloneqq \pi_M(EL_M)$. In order for the classical theory to be well defined one usually requires the projection $L_{\partial M\times[0,\epsilon]}$ of the critical locus (associated to a suitable collar) to be a Lagrangian submanifold of the space of boundary fields $F^\partial_{\partial M}$. This is to allow the existence of solutions to the Cauchy problem after a choice of a suitable boundary condition: another Lagrangian submanifold $L\subset F_{\partial M}$ that should be transversal to $L_M$. 

If we further define $C_{\partial M}\subset F_{\partial M}^\partial$ to be the space of pairs of boundary fields that can be completed to a solution in $L_{\partial M\times[0,\epsilon]}$, we can conclude that $C_{\partial M}$ must be a coisotropic submanifold \cite{CMRCorfu}.

In this paper we will analyse the submanifold $C_{\partial M}$ associated with Palatini--Cartan--Holst theory of gravity as given by the vanishing locus of some local functionals on the space of fields. One could then run the Batalin--Fradkin--Vilkovisky (BFV \cite{BFV1,BFV2}) machinery to replace the reduction of said coisotropic submanifold with the associated BFV-complex \cite{Schaetz09, Schaetz10}. Moreover, one could perform a Batalin-Vilkovisky (BV \cite{BV81}) analysis the PCH theory in the bulk and study whether the BV-BFV axioms are satisfied \cite{CMR1}. This will be done by the authors in \cite{CS2}.

\section{General Relativity in the Palatini--Cartan--Holst formalism}\label{Sect:tetrad}
The Einstein--Hilbert theory of gravity does not strictly look like a gauge theory, as it is not manifestly a theory of connections, unlike electromagnetism or chromodynamics or the standard model of particle physics. Nevertheless, there is a classically-equivalent formulation of GR as a theory of principal connections, an action functional that produces a set of equations of motion that can be reduced to Einstein's equations, and yet it is different from a structural point of view. 

The setting is as follows, consider the principal fibre bundle of (co-) frames on $M$, with the natural action of $\mathsf{SO}(3,1)$ on it\footnote{We assume that $M$ is orientable and can carry a Lorentzian structure.}. The dynamical fields are the co-frame field $e\colon TM\stackrel{\sim}{\longrightarrow} \mathcal{V}$ (also called vierbein, in four dimensions), with $\mathcal{V}$ a vector bundle on $M$ with a reference smooth choice of Minkowski metric $\eta$  and of orientation on each fibre $V$, and a connection $\omega$ on the principal $\mathsf{SO}(3,1)$ bundle $P$ associated to $\mathcal{V}$, i.e. locally on an open $U\subset M$, $\omega|_U\colon U\longrightarrow \mathfrak{so}(3,1)$. Note that we require $e$ to be an isomorphism. 

\begin{remark}

Generalisations of this construction have been considered (for instance) by Floreanini and Percacci \cite{FloPer}, where the internal group is $\mathsf{GL}(4)$ and no compatibility is required between a principal connection $\omega$ and either the co-frame $e$ (torsionless condition) or an internal metric $\kappa$ (metricity condition), which is to be interpreted as a dynamical generalisation of the Minkowski metric. A metric field $g$ and a connection $A$ on the tangent bundle are then obtained by pulling back the internal metric and connection respectively. This construction reduces to the usual Palatini--Cartan approach by requiring either the metricity or torsionless conditions and recovering the other one through the field equations.
\end{remark}

One can consider $\bigwedge^2V$-valued local connections, using the isomorphism
\begin{equation}\label{e:SPCH}
\Wedge{2}V\stackrel{\sim}{\longrightarrow}\mathfrak{so}(3,1)
\end{equation}
induced by $\eta$, which maps the basis $e_i\wedge e_j$ to a basis of matrices $t^i_j$ of the Lie algebra by raising/lowering indices. In this setting, the theory is fully described by the \emph{Palatini--Cartan action functional}:
\begin{equation}\label{Pal}
S_{\text{PC}}=\int\limits_M \mathrm{Tr}\left[ e\wedge e \wedge F_\omega + \frac{\Lambda}{4} e^4\right],
\end{equation}
where by $\mathrm{Tr}\colon \bigwedge^4 V\longrightarrow \mathbb{R}$ we denote the volume form in $\bigwedge^4 V$ normalised such that $\mathrm{Tr}(u_i\wedge u_j \wedge u_k\wedge u_l) = \epsilon_{ijkl}$, with $\{u_i\}_{i=1}^4$ an $\eta$-orthonormal basis in $V$, and $\Lambda$ is the cosmological constant\footnote{Note that the choice of orientation of $V$ and the orientability of $M$ are not really necessary, as a top form with values in $\Wedge{4}V$ is a density and can be integrated on $M$ without any further choice.}.  

To fix the notation we will assume throughout that $d_\omega\phi=d\phi + [\omega,\phi]$ where $\phi$ is any $\mathfrak{g}$-module valued field (i.e. we use the bracket notation to denote \emph{any} Lie-algebra action), and the curvature form reads $F_\omega=d\omega + \frac12[\omega,\omega]$.

The Euler--Lagrange equations for the associated variational problem yield, at the same time, Einstein's equation and the compatibility condition between $\omega$ and $e$. The latter condition imposes that covariant derivatives be taken w.r.t. the Levi--Civita connection. More explicitly, the Euler--Lagrange equations for \eqref{Pal} read
\begin{subequations}\label{classconstraints}\begin{align}\label{eqhalfshell}
d_\omega(e\wedge e)  & = 0 \\\label{Einsteq}
e\wedge F_\omega + \Lambda e^3  & =0. 
\end{align}\end{subequations}
Observe that Equation \eqref{eqhalfshell} is equivalent to $d_\omega e=0$ in the bulk, while equation $e\wedge F_\omega=0$ is rewritten as $\sum_i[F_\omega]_{\mu \nu}^{\mu\rho}=0$ with $\mu,\nu,\rho =1\dots 4$ indices of the basis $\{e_\mu\}$. 
\begin{remark}
Notice that we assume that $e$ is an isomorphism, and in this case Equations \eqref{classconstraints} describe the same geometro-dynamics of the Einstein-Hilbert variational problem, up to gauge equivalence. Indeed, Eq. \eqref{Einsteq} is Einstein's field equation for the metric $g=e^*\eta$ when $\omega=\omega(e)$, the unique solution of \eqref{eqhalfshell}, is understood to be the Levi--Civita connection. As a matter of fact, since we require $\omega$ to be $\eta$-compatible, the torsionless condition $d_\omega e=0$ will imply the metricity condition $d_{e^*\omega} g=0$, which is uniquely solved by the Levi--Civita metric connection.
\end{remark}

\begin{remark}Observe that the map $e\wedge\cdot$ in $\Omega^\bullet(M,\Wedge{\bullet} \mathcal{V})$ is not necessarily an isomorphism, even if $e$ is nondegenerate. As a matter of fact $e\wedge F_\omega=0$ is not equivalent to the flatness condition $F_\omega=0$.
\end{remark}

Strictly speaking, EH and PCH theories are equivalent only when condition \eqref{eqhalfshell} is used to rewrite the Palatini--Holst action in terms of the curvature of Levi--Civita connection. However, the way one encodes the \emph{half-shell} constraint \eqref{eqhalfshell} is somehow arbitrary. As a matter of fact, adding an explicit Lagrange multiplier will not modify the equations of motion (Theorem \ref{Theo:halfshellcl}), but will have a non-trivial effect on the boundary. We will consider this option in Section \ref{clhalfshell}.

The \emph{minimality} of the theory has been discussed, for instance, in \cite{Per}, where it is shown how one can easily consider the most general theory of gravity of this kind to be a \emph{topological}\footnote{The term topological being referred to the fact that it does not affect the dynamics.} modification of the Palatini--Cartan action. This modification goes under the name of Holst action \cite{Holst}, and it is still possible to add a finite number of boundary corrections to it. We report here the shape of such general action for completeness:

\begin{align}\label{Holstgeneral}
S_{\text{tot}}=&\int\limits_{M}\mathrm{Tr}\left[\alpha_1(e\wedge e\wedge F_\omega) + \alpha_2\star(e\wedge e)\wedge F_\omega \right] + \alpha_6(\Lambda) \mathrm{Tr}(e^4)\\\notag
 +& \int\limits_M (\alpha_3 - i\alpha_4) dL_{CS}(\omega^-) + (\alpha_3+i\alpha_4)dL_{CS}(\omega^+) + \alpha_5 d(d_\omega \star e\wedge e).
\end{align}
A few comments are in order. The \emph{trace} is again induced by the volume form in $V$ and we used the internal Hodge $\star$ for Minkowski metric $\eta$. The $\alpha_1,\alpha_2$ terms, with respect to a basis $\{u_i\}_{i=1}^4$ explicitly read: 
\begin{equation}\label{genvol}
\left(\alpha_1 \epsilon_{ijkl} e^i\wedge e^j\wedge F^{kl} + \alpha_2 e^i\wedge e^j\wedge F^{kl} \eta_{ik}\eta_{jl}\right) \in \Omega^{\text{top}}(M),
\end{equation}
with $\eta_{ij}$ the Minkowski metric, diagonal $\eta=\mathrm{diag}\{1,1,1,-1\}$ in the basis $\{u_i\}_{i=1}^4$. Eq. \eqref{genvol} will be interpreted in Lemma \ref{Lem:twistpairing} as a volume form in the top exterior power $\Wedge{4} V$.

The coefficient $\alpha_6(\Lambda)$ is proportional to the cosmological constant, whereas the components $\omega^{\pm}$ are respectively the (anti-)self-dual parts of the connection $\omega$ (with respect to the Hodge dual) and the functionals $L_{CS}$ are Chern-Simons forms. It can be seen \cite{Per}, that the total derivative terms in \eqref{Holstgeneral} unfold to yield topological terms proportional to the Pontrjagin, Euler and Nieh--Yan classes. 

Notice that the terms from $\alpha_3$ to $\alpha_5$ are relevant neither for the dynamical theory nor for the boundary structure. As a matter of fact they arise as exact corrections to the boundary 1-form, and therefore they will not change the symplectic structure.

\subsection{Palatini--Cartan--Holst theory}\label{Holstname}
The $\alpha_2$ term in \eqref{Holstgeneral} will have a non trivial effect in both bulk and boundary theories, and we shall retain it in what follows. The other topological boundary terms will be discarded in this analysis. In doing this we will rename our parameters as it is customary in the literature, namely by introducing the so-called Barbero--Immirzi \cite{Bar,Imm} parameter $\gamma\in\mathbb{R}\backslash\{0\}$ and considering the (real) Holst action
\begin{equation}\label{Holstaction}
S_{\text{Holst}}=\int\limits_{M}\mathrm{Tr}\left(e\wedge e\wedge F_\omega + \frac{1}{\gamma}\star(e\wedge e)\wedge F_\omega  + \alpha_6(\Lambda) e^4\right).
\end{equation}
This theory reduces to the Palatini--Cartan action in the limit $\gamma\rightarrow \infty$, but it still describes the same (Einstein) equations, up to a rescaling factor $\gamma$. However, this apparently harmless shift turns out to be a source of \emph{ambiguity} in the quantisation scheme \cite{Imm,Rovth}.

First introduced by Barbero \cite{Bar} to generalise the construction of Ashtekar canonical quantum gravity \cite{Asht} in terms of a real $SU(2)$ connection, this parameters was later studied by Immirzi \cite{Imm}. Ashtekar's formulation dealt with complex self-dual connections instead, which are recovered by complexifying the action and setting $\gamma=\pm i$. This complexification is avoided with the introduction of a $\gamma$-dependent canonical transformation of the phase space, mapping the Palatini fields to some $\gamma$-rescaled fields. This parameter dependence has been observed to be \emph{non-quantisable} \cite{Rovth}, in the sense that it cannot be unitarily implemented, which means that the quantisation of the theory without the $\gamma$ parameter is not  unitarily equivalent to the scaled one. However, we will be interested in the semiclassical structure only, and the generalisation introduced by the Barbero--Immirzi parameter will be taken into account only for completeness.

The introduction of the $\gamma$ parameter changes the pairing structure between $e\wedge e$ and $F_\omega$. To understand this statement we need the following:

\begin{lemma}\label{starcyclic}
Let $A,B\in \Wedge{2}V\simeq\mathfrak{so}(\eta)$, with $\eta$ a nondegenerate metric on a four dimensional vector space $V$. Then
$$
\star[A,B] = [\star A, B] = [A,\star B]. 
$$
\end{lemma}

\begin{proof}
Consider $A,B,C\in \Wedge{2}V$. Denote by $\mathrm{Tr}$ the invariant volume form\footnote{$\mathrm{Tr}$ is given by $\epsilon_{ijkl}$ in the $\eta$-orthonormal basis} in $\Wedge{4}V$, and by $\tilde{A},\tilde{B},\tilde{C}$ the respective matrices in $\mathfrak{so}(\eta)$. We compute\footnote{We consider $\det{\eta}=\pm1$ for simplicity.}:
\begin{multline}
\mathrm{Tr}\left[ C\star[A,B]\right]=\frac12 \epsilon_{ijkl} C^{ij} \epsilon^{kl}_{rs} A^{rm} B^{ns} \eta_{mn}  \\
=2 C^{ij}A^{rm}B^{ns}\eta_{i\langle r}\eta_{s\rangle j}\eta_{mn} = -2 \mathrm{Tr}_{\mathfrak{so}(\eta)}[\tilde{B}\cdot \tilde{C}\cdot \tilde{A}] = \mathcal{T},
\end{multline}
(angular brackets denote antisymmetrisation of enclosed indices) showing that $\mathrm{Tr}\left[ C\star[A,B]\right]$ is cyclic in $A,B,C$. Then, since $\mathrm{Tr}[A\star B] = \mathrm{Tr}[\star A B]$ we have
$$\mathcal{T} = \mathrm{Tr}\left[ A\star[B,C]\right]=\mathrm{Tr}[C [\star A, B]]=\mathrm{Tr}\left[ C\star[A,B]\right],$$
and
$$\mathcal{T} = \mathrm{Tr}\left[ B\star[C,A]\right]=\mathrm{Tr}[C [A, \star B]]=\mathrm{Tr}\left[ C\star[A,B]\right].$$
The identities are valid for all $C\in \Wedge{2}V$, and the volume form $\mathrm{Tr}$ induces a nondegenerate pairing on $\Wedge{2}V$, proving the claim.
\end{proof}

\begin{lemma}\label{Lem:twistpairing}
Consider the complexification of Minkowski vector space $(V_{\mathbb{C}},\eta)$ and the maps (we drop the subscript $\mathbb{C}$)
\begin{align}
T_\gamma\colon&\begin{array}[t]{ccc} 
	\bigwedge^2 {V} & \longrightarrow & \bigwedge^2 {V} \\
	\alpha &\longmapsto & \alpha + \frac{1}{\gamma} \star \alpha
\end{array}\\
\hat{T}_\gamma\colon&\begin{array}[t]{ccc}
	\Wedge{2}V\otimes \Wedge{2}{V} & \longrightarrow & \mathbb{C}\\
	\alpha\otimes\beta &\longmapsto&\mathrm{Tr}[{T}_\gamma(\alpha)\wedge\beta ]
\end{array}\\
\widetilde{T}_\gamma\colon&\begin{array}[t]{ccc}
	\bigwedge^2{{V}} & \longrightarrow & \bigwedge^2{V}^*\\
	\alpha & \longmapsto & \hat{T}_\gamma(\alpha\otimes \cdot)
\end{array}\end{align}
for $\gamma\in \mathbb{C}\backslash\{0\}$. Then, $T_\gamma$ and $\widetilde{T}_\gamma$ are isomorphisms for any value $\gamma\not=\pm i$, and $\hat{T}_\gamma$ defines a non-degenerate symmetric inner product in $\bigwedge^2{V}$. Moreover, ${T}_\gamma$ is symmetric with respect to the inner product induced by the volume form in $\Wedge{4}V$, i.e. 
$$\mathrm{Tr}[{T}_\gamma(\alpha)\wedge\beta ] = \mathrm{Tr}[\alpha\wedge {T}_\gamma(\beta)],$$
and 
\begin{equation}
T_\gamma[\alpha,\beta] = [T_\gamma[\alpha],\beta] = [\alpha,T_\gamma[\beta]].
\end{equation}
\end{lemma}

\begin{proof}
Consider the linear map $\hat{T}_\gamma\colon \bigwedge^4 V \longrightarrow \mathbb{R}$ and evaluate it on a basis $\{u_i\wedge u_j\}$, where $\{u_i\}$ is an orthonormal basis for $(V,\eta)$:
\begin{equation}\label{hatTgamexp}
\hat{T}_\gamma[u_i\wedge u_j \otimes u_k\wedge u_l]=\left[\epsilon_{ijkl} + \frac2\gamma\eta_{i\langle k}\eta_{l\rangle j}\right],
\end{equation}
as it can be easily checked using $\star (u_i\wedge u_j) = \frac12 \epsilon_{ijkl}\eta^{k\langle m}\eta^{n\rangle l}u_m\wedge u_n$ and the normalisation $\mathrm{Tr}(u_i\wedge u_j \wedge u_k\wedge u_l)=\epsilon_{ijkl}$ (angular brackets denote antisymmetrisation of enclosed indices). If we relabel the basis indices in $\bigwedge^2{V}$ as $(12,13,14,23,24,34)\rightarrow (\mathbf{1},\mathbf{2},\mathbf{3},\mathbf{4},\mathbf{5},\mathbf{6})$. It is simple to gather that the representative matrix of $\widetilde{T}_\gamma$ with respect to the canonical bases in $\bigwedge^2 V$ and $\bigwedge^2 V^*$, relabelled as just mentioned, is given by
$$[\widetilde{T}_\gamma]=\left(
\begin{array}{cccccc}
 \gamma^{-1} & 0 & 0 & 0 & 0 & 1 \\
 0 & \gamma^{-1} & 0 & 0 & -1 & 0 \\
 0 & 0 & -\gamma^{-1} & 1 & 0 & 0 \\
 0 & 0 & 1 & \gamma^{-1} & 0 & 0 \\
 0 & -1 & 0 & 0 & -\gamma^{-1} & 0 \\
 1 & 0 & 0 & 0 & 0 & -\gamma^{-1} \\
\end{array}
\right)$$
and its determinant is $\mathrm{det}[\widetilde{T}_\gamma] = -(1+\gamma^{-2})^3$. Now, the combination 
\[
f^\alpha_{ij}:=u_i\wedge u_j + \alpha\ \eta_{im}\eta_{jn}\epsilon^{mnkl}u_k\wedge u_l,
\]
for $\alpha\in\mathbb{R}$, is a basis of $\bigwedge^2 V$ for all $\alpha\not=\pm i$. In fact, the linear map $F_\alpha$ mapping $\{u_i\wedge u_j\}$ to $\{f^\alpha_{ij}\}$ reads
$$[F_\alpha]=\left(
\begin{array}{cccccc}
 1 & 0 & 0 & 0 & 0 & \alpha \\
 0 & 1 & 0 & 0 & -\alpha & 0 \\
 0 & 0 & 1 & -\alpha & 0 & 0 \\
 0 & 0 & \alpha & 1 & 0 & 0 \\
 0 & \alpha & 0 & 0 & 1 & 0 \\
 -\alpha & 0 & 0 & 0 & 0 & 1 \\
\end{array}
\right)$$
and $\mathrm{det}(F_\alpha)=(1+\alpha^2)^3$. In particular, for $\alpha=\frac\gamma2$ we have ${T}_\gamma\equiv F_{\frac{\gamma}{2}}$.

To prove the symmetry of ${T}_\gamma$ we can compute
\begin{multline*}
\mathrm{Tr}[{T}_\gamma(\alpha)\wedge\beta]=\alpha^{ij}\beta^{mn}\mathrm{Tr}\left[\left(u_iu_j + \frac{1}{2\gamma}\epsilon_{ij}^{kl}u_ku_l\right)\wedge u_mu_n\right]=\\
=\alpha^{ij}\beta^{mn}\left(\epsilon_{ijmn} +\frac{1}{2\gamma}\epsilon_{ij}^{kl}\epsilon_{klmn}\right)=\alpha^{ij}\beta^{mn}\left(\epsilon_{ijmn} +\frac{1}{2\gamma}\epsilon_{ijpq}\epsilon^{pq}_{mn}\right)=\\
=\alpha^{ij}\beta^{mn}\mathrm{Tr}\left[u_iu_j\wedge\left(u_mu_n + \frac{1}{2\gamma}\epsilon_{mn}^{pq}u_pu_q\right)\right] = \mathrm{Tr}[\alpha\wedge{T}_\gamma(\beta)],
\end{multline*}
or equivalently use the fact that $\mathrm{Tr}[{T}_\gamma(\alpha)\wedge\beta]=\hat{T}_\gamma(\alpha\otimes\beta)$ and that $\hat{T}_\gamma$ is a manifestly symmetric bilinear map (c.f. \eqref{hatTgamexp}) on $\bigwedge^2 V$. Finally, using Lemma \ref{starcyclic} we can easily check that $T_\gamma[\alpha,\beta] = [T_\gamma[\alpha],\beta] = [\alpha,T_\gamma[\beta]]$.
\end{proof}

\begin{proposition}\label{Tgammamorph}
The map $T_\gamma \colon \Wedge{2} V \longrightarrow \Wedge{2} V$ is a Lie algebra morphism of $\mathfrak{so}(\eta)$ if and only if $\gamma^2=s$, with $s$ the sign of the determinant of $\eta$. Then:
\begin{equation}
T_{\pm\sqrt{s}}[A,B] = \frac12[T_{\pm\sqrt{s}}[A],T_{\pm\sqrt{s}}[B]].
\end{equation}
\end{proposition}

\begin{proof}
With the help of Lemma \ref{starcyclic} we compute
$$[T_\gamma[A], T_\gamma[B]] = [A,B] + \frac{2}{\gamma} \star [A,B] + \frac{1}{\gamma^2} [A,\star^2 B] = (1 + \frac{s}{\gamma^2}) [A,B]  + \frac2\gamma \star [A,B].$$
We can ask for $[T_\gamma[A], T_\gamma[B]] = c T_\gamma [A,B]$, imposing the conditions
\begin{align*}
\frac{c}{\gamma}&=\frac{2}{\gamma},\\
1 + \frac{s}{\gamma^2} &= c,
\end{align*}
which immediately show that this is possible only when $\gamma^2 = s$, and $c=2$ in that case. 
\end{proof}

\begin{remark}
In the present paper we will assume $\gamma\not=\pm i$, in order to have $T_\gamma$ nondegenerate (cf. Lemma \ref{Lem:twistpairing}). As mentioned, one can make sense of PCH theory in the complexification of $\mathfrak{so}(3,1)$, relating to the Ashtekar formulation of Palatini gravity \cite{Asht}, and fixing $\gamma=\pm i$ allows self-dual connections (in the complexified Lie algebra). Thus, $T_{\pm i}$ loses invertibility but gains the property of being a Lie algebra morphism.  From now on, however, we will consider $\gamma\in\mathbb{R}$. 
\end{remark}

We can summarise the previous constructions by giving the following definitions we will use throughout. Let $P\longrightarrow M$ the principal $SO(3,1)$ bundle associated to $\mathcal{V}$, $\mathcal{A}_P$ the space of principal connections, and denote by $\Omega_{nd}^1(M,\mathcal{V})$ the space of (nondegenerate) tetrads. 

\begin{definition}\label{Def:PCH}
The assignment of the pair $(\mathcal{F}_{PCH}, S_{PCH})_M$ to every $4$ dimensional manifold $M$ with
\begin{equation}
\mathcal{F}_{PCH}=\Omega_{nd}^1(M, \mathcal{V}) \times \mathcal{A}_P \ni (e,\omega),
\end{equation}
and
\begin{equation}\label{Action:PCH}
S_{PCH}=\intl_M \hat{T}_\gamma\left[ \frac12 e\wedge e\wedge F_\omega +\frac{\Lambda}{4} e^4 \right].
\end{equation}
will be called \emph{Palatini--Cartan--Holst theory}.
\end{definition}

\begin{definition}\label{Def:HSPCH}
The assignment of the pair $(\mathcal{F}_{HS}, S_{HS})_M$ to every $4$ dimensional manifold $M$ with
\begin{equation}
\mathcal{F}_{HS}=\Omega_{nd}^1(M, \mathcal{V}) \times \mathcal{A}_P \times \Omega^2(M,\Wedge{3}\mathcal{V}) \ni (e,\omega, t),
\end{equation}
and
\begin{equation}\label{Action:HSPCH}
S_{PCH}=\intl_M \hat{T}_\gamma\left[ \frac12 e\wedge e\wedge F_\omega + \frac{\Lambda}{4} e^4 \right] + \mathrm{Tr} \left[t\wedge d_\omega e \right].
\end{equation}
will be called \emph{Half-Shell Palatini--Cartan--Holst theory}.
\end{definition}

\begin{remark}
We understand the integrand $\hat{T}_\gamma\left[ \frac12 e\wedge e\wedge F_\omega +\frac{\Lambda}{4} e^4 \right]$ as the pairing in $\Wedge{2}V$ between (the vector values of) $e\wedge e$ and $T_\gamma[F_\omega]$. This is equivalent to $\mathrm{Tr}\left[e\wedge e\wedge T_\gamma[F_\omega] + \frac\Lambda4 e^4\right]$ since 
$$ee\wedge\star [ee]=\eta_{i\langle k}\eta_{l\rangle j}e^i\wedge e^j\wedge e^k\wedge e^l=-\eta_{i\langle k}\eta_{l\rangle j}e^k\wedge e^j\wedge e^i\wedge e^l=0.$$
\end{remark}

\section{Classical boundary structure}\label{Sect:Classbound}
In this section we will analyse the structure that is induced on the boundary $\partial M$ by the bulk Palatini--Cartan--Holst theory (Sections \ref{ClassPalHol} and \ref{PCHCA}) and by the bulk Half-Shell Palatini--Cartan--Holst theory, a modification of it where the compatibility constraint $d_\omega e$ is enforced by means of a Lagrange multiplier (Section \ref{clhalfshell}).

\subsection{PCH boundary structure}\label{ClassPalHol}
Recalling Definition \ref{Def:PCH}, the space of physical fields for the PCH theory of gravity is given by $\mathcal{F}_{PCH}=\Omega^1_{nd}(M, \mathcal{V})\times \mathcal{A}_P$. A connection is locally described by a (local) one-form $\omega$ (on a chart) with values in $\mathfrak{so}(3,1)\simeq\bigwedge^2 V$. The boundary inclusion $\iota\colon \partial M \longrightarrow M$ induces the bundles $P^\partial:=\iota^*P$ and $\mathcal{V}^\partial=\iota^*\mathcal{V}$ over $\partial M$. We denote by $\Omega^1_{nd}(\partial M,\mathcal{V}^\partial)$ the space of $V$-valued one-forms on the boundary that span a three-dimensional subspace of $V$ at each point, with $\mathcal{V}^\partial := \iota^*\mathcal{V}$.

\begin{remark}
In the literature (e.g. \cite{Asht,Rovth}), globally hyperbolic structure of space-time is usually assumed for Palatini--Holst gravity. We will instead begin by considering any $3+1$-dimensional manifold with boundary, without specifying the kind of boundaries we allow. This means we will not put any extra restriction on the fields (compare with \cite{CS1}, where the Einstein--Hilbert action for GR in the ADM formalism  is analysed). In the analysis of the reduced phase space we will however assume that the co-frame field is such that the metric it induces on the boundary is nondegenerate (note that this is an open condition on the space of bulk fields).
\end{remark}

Consider the following
\begin{lemma}
\label{Lemma:kerWe}
The map 
$$\mathsf{W}_e^{(p,k)}\colon \Omega^p\left(\partial M,\bigwedge^{k}\mathcal{V}^\partial\right)\longrightarrow  \Omega^{p+1}\left(\partial M,\bigwedge^{k+1}\mathcal{V}^\partial\right),$$
defined by $\mathsf{W}^{(p,k)}_e(X)=X\wedge e$, where $e$ is the restriction of the tetrad to the boundary $\iota\colon \partial M \rightarrow M$, is injective for $p=k=1$ and it is surjective when $(p,k)=(1,2)$ or $(p,k)=(2,1)$. 
\end{lemma}
\begin{proof}
We use the fact that $\{e_\mu\}$ is a basis of $V$ and expand the component $X_a\in \bigwedge^2 V$ in that basis. 
When $(p,k)=(1,2)$ the kernel is determined by the equation
\begin{equation}
X\wedge e = X_a^{\mu\nu}e_\mu e_\nu e_b dx^a\wedge dx^b=0.
\end{equation}
This yields a system of equations as follows
\begin{eqnarray*}
(1,2) & X_1^{nc}e_n e_c e_2 - X_2^{nc} e_n e_c e_1 =0 & X_1^{13}e_1 e_3 e_2 - X_2^{23}e_2 e_3 e_1=0\\
(1,3)& X_1^{nc}e_n e_c e_3 - X_3^{nc} e_n e_c e_1 =0 & X_1^{12}e_1 e_2 e_3 - X_3^{32}e_3 e_2 e_1=0\\
(2,3)& X_2^{nc}e_n e_c e_3 - X_3^{nc} e_n e_c e_2 =0 & X_2^{21}e_2 e_1 e_3 - X_3^{31}e_3 e_1 e_2=0\\
\end{eqnarray*}
from which we easily obtain $\sum_{a=1}^3X_a^{ab}=X_a^{nb}=0$, since $A+B=A+C=B+C=0$ implies $A=B=C=0$ (check the resulting equation for $e_ne_1e_2$ in the first line). However, we can look at the inhomogeneous equation 
$$t=X\wedge e,$$
and expand $t$ in the basis of the $e_\mu$'s: $t=t_{ab}^{\mu\nu\rho}e_\mu e_\nu e_\rho dx^adx^b$, and again for $X\wedge e$. With a similar argument one obtains that the equation is solved for $X$ as
\begin{align}
X^{ab}_a & = - t^{\hat{b}df}_{df}\epsilon_{\hat{b}df} \\ 
2 X_{\hat{a}}^{\hat{a}n} & =-t^{n\hat{a}b}_{\hat{a}b} + t^{ncd}_{cd}\epsilon_{\hat{a}cd},
\end{align}
where hatted indices are not summed over and, because we know that combinations such that $\sum_a X_a^{ab}=0$ are in the kernel of $e\wedge\cdot$, this proves surjectivity of the map.

Again, with an essentially identical argument we can compute the kernel of $\mathsf{W}_e^{(2,1)}$ to be $\sum_aX_{ab}^a=X^n_{ab}=0$, and show that is is surjective on $\Omega^3(\partial M,\Wedge{2}\mathcal{V}^\partial)$.

On the other hand, when $(p,k)=(1,1)$ we have
\begin{equation}
X_{\langle a}^\mu e_\mu e_{b\rangle} = 0,
\end{equation}
which yields
\begin{eqnarray*}
(1,2) & X_1^ne_ne_2 + X_1^3 e_3 e_2 - X_2^n e_n e_1 - X_2^3e_3e_1 +(X_1^1+X_2^2)e_1e_2 =0\\
(1,3) & X_1^ne_ne_3 + X_1^2 e_2 e_3 - X_3^n e_n e_1 - X_3^2e_2e_1 +(X_1^1+X_3^3)e_1e_3 =0\\
(2,3) & X_2^ne_ne_3 + X_2^1 e_1 e_3 - X_3^n e_n e_2 - X_3^1e_2e_3 +(X_2^2+X_3^3)e_2e_3 =0
\end{eqnarray*}
and one can infer that $X=0$.
\end{proof}

\begin{remark}[Definition of splitting]\label{dualsplitting}
We can consider the choice of a complement of the kernel $\mathcal{W}_{(i,j)}=\mathrm{ker}\mathsf{W}_e^{(i,j)}$, i.e. $\Omega_{(i,j)}\coloneqq\Omega^i(\partial M,\Wedge{j}\mathcal{V}^\partial)=\mathcal{W}_{(i,j)}\oplus\mathcal{C}_{(i,j)}$ for $(0\leq i \leq 3, 0\leq j \leq 4)$. 
Note that $C_{(i,j)}$ is actually a complement for all $e$ in an open neighbourhood of a given $e_0$ inside $\Omega_{nd}(\partial M,\mathcal{V}^\partial)$. The dual space is 
$$\Omega_{(i,j)}^*\equiv\left(\Omega^i(\partial M, \Wedge{j}\mathcal{V}^\partial)\right)^*\simeq \Omega^{3-i}(\partial M, \Wedge{4-j}\mathcal{V}^\partial),$$
and we can consider the annihilator $\mathcal{W}_{(i,j)}^0=\mathrm{Im}(\mathsf{W}^{(2-i,3-j)})$. Then the dual splitting reads
\begin{equation}
\Omega_{(i,j)}=\mathcal{W}_{(i,j)}\oplus \mathcal{C}_{(i,j)};\ \ \ \Omega_{(i,j)}^*=\mathcal{W}_{(i,j)}^*\oplus \mathcal{W}_{(i,j)}^0;
\end{equation}
and again $\mathcal{C}_{(i,j)}^*=\mathcal{W}_{(i,j)}^0=\mathrm{Im}(\mathsf{W}_e^{(2-i,3-j)})$. In particular, for $(i,j)=(1,2)$, we have that an object in the dual of the complement of $\mathrm{ker}\mathsf{W}_e^{(1,2)}$ is in the image of $\mathsf{W}_e^{(1,1)}$. 
\end{remark}

\begin{remark}[Definition of projections]\label{projremark}
Once we fix a complement $\mathcal{C}_{(i,j)}$ of the kernel $\mathcal{W}_{(i,j)}$ we can define projections $p_{(i,j)}\colon\Omega_{(i,j)}\longrightarrow \mathcal{W}_{(i,j)}$ and $p_{(i,j)}'\colon \Omega_{(i,j)} \longrightarrow \mathcal{C}_{(i,j)}$. In what follows, we will drop the subscripts and simply denote the projections $p$ and $p'$.  Dually, we will generically call $p^\dag$ the projection to the image of $W_e$. Observe that all $p$, $p'$ and $p^\dag$ will depend on $e$.
\end{remark}

\begin{remark}
Notice that we can use the inverse function theorem in the Banach manifold of (fixed regularity type) sections to show that both $\mathcal{W}_{(i,j)}$ and its complement $\mathcal{C}_{(i,j)}$ (which can be seen as the annihilator of the kernel $\mathcal{W}_{(2-i,3-j)}$) are subbundles, owing to the fact that they are constant rank.
\end{remark}

\begin{theorem}\label{Theo:PCHClassical}
The classical space of boundary fields for Palatini--Cartan--Holst theory is the symplectic manifold given by the fibre bundle
\begin{equation}
\mathcal{F}_{PCH}^{\partial}\longrightarrow \Omega^1_{nd}(\partial M,\mathcal{V}^\partial),
\end{equation} 
with fibre over $e\in\Omega^1_{nd}(\partial M,\mathcal{V}^\partial)$ given by the reduction $\mathcal{A}_{\iota^*P}^{red}\coloneqq\qsp{\mathcal{A}_{\iota^*P}}{\sim}$, with respect to the equivalence relation $\omega \sim \omega' \iff \omega- \omega' \in \mathrm{Ker}(\mathsf{W}^{(1,2)}_e)$, where $\mathcal{V}^\partial=\iota^*\mathcal{V}$ and $\iota\colon \partial M \longrightarrow M$, and symplectic form depending on the Barbero--Immirzi parameter $\gamma$:
\begin{equation}\label{classicalboundaryform}
\varpi^{\partial}_{PCH}=-\int\limits_{\partial M}\hat{T}_\gamma\left[\mathbf{e}\wedge\delta\mathbf{e}\wedge\delta\bom \right].
\end{equation}
The surjective submersion $\pi_{PCH}\colon \mathcal{F}_{PCH}\longrightarrow \mathcal{F}_{PCH}^{\partial}$ has the explicit expression:
\begin{equation}\label{classprojection}
\pi_{PCH}\colon\begin{cases}
\mathbf{e}=e\\
\bom = [\omega]_e, 
\end{cases}
\end{equation}
with $\bem\in\Omega_{nd}^1(\partial M,\mathcal{V}^\partial)$ and $[\omega]_e\in \mathcal{A}_{\iota^*P}^{red}$. 

Furthermore, there is a symplectomorphism $\mathcal{F}^\partial_{PCH}\longrightarrow T^* \Omega^1_{nd}(\partial M,\mathcal{V}^\partial)$ by means of the identification 
$$-\bem\wedge T_\gamma[\bom] =:\mathbf{t}_\gamma \in \Omega^2(\partial M,\Wedge{3}\mathcal{V}^\partial)\simeq\mathcal{A}_{\iota^*P}^{red},$$
and the symplectic form reads
\begin{equation}
\varpi^{\partial}_{PCH}=\intl_{\partial M}\mathrm{Tr}\left[\delta\bem\wedge\delta\mathbf{t}_\gamma\right].
\end{equation}

\end{theorem}

\begin{proof}
The variation of the Palatini--Cartan--Holst action \eqref{Action:PCH} splits into a bulk term, which we will not consider in what follows, and a boundary term. The latter is interpreted as a one-form on the space of pre-boundary fields $\widetilde{\mathcal{F}}_{PCH}$ of restriction of fields and normal jets to the boundary, and it reads 
\begin{equation}\label{holst1form}
\widetilde{\alpha}_{PCH}=-\frac12\int\limits_{\partial M}  \hat{T}_\gamma\left[e\wedge e \wedge \delta \omega \right].
\end{equation}
The pre-boundary two-form $\widetilde{\varpi}_{PCH}\coloneqq\delta\widetilde{\alpha}_{PCH}$ then reads
\begin{equation}\label{holst2form}
\widetilde{\varpi}_{PCH}=-\int\limits_{\partial M}  \hat{T}_\gamma\left[\delta e\wedge e \wedge \delta \omega\right].
\end{equation}
The restrictions of fields to the boundary are denoted with the same symbols, but we understand $\omega$ as an $\mathfrak{so}(3,1)$-connection on the induced principal bundle $P_{\partial M}=\iota^*P_M$ on the boundary $\iota\colon \partial M \hookrightarrow M$, while $e$ is a $V$-valued one-form on the boundary (i.e. $e\in\Omega_{nd}^1(\partial M,\iota^*\mathcal{V})$). In the basis $\{u_i\}_{i=1\dots 4}$ of $V$ we have $e=e_a^iu_idx^a$ whereas $\omega=\omega_a^{ij}u_i\wedge u_jdx^a$ where we fix that the indices $a,b,c$ run over the boundary directions $1,2,3$. Notice, however, that the vectors $e_a=e_a^iu_i$ are a basis of a three dimensional subspace $W\subset V$, and we can complete it to a basis of $V$ by introducing a vector $e_n$, linearly independent from the $e_a$'s.

Using Lemma \ref{Lem:twistpairing} 
we can read the equations defining the kernel of $\widetilde{\varpi}_{PCH}$ from \eqref{holst2form}
\begin{subequations}\begin{align}\label{XeHolst}
(X_e)\wedge e &= 0\\\label{XomegaHolst}
T_\gamma[\X{\omega}]\wedge e&=0.
\end{align}\end{subequations}
In virtue of Lemma \ref{Lemma:kerWe}, this means that $\mathrm{ker}(\widetilde{\varpi}_{PCH})= \mathrm{ker}(\mathsf{W}^{(1,1)}_e)\times \mathrm{ker}(\mathsf{W}^{(1,2)}_e)= \mathrm{ker}(\mathsf{W}^{(1,2)}_e)$, since we can read (and solve) \eqref{XeHolst} as
\begin{equation}
\mathsf{W}^{(1,1)}_e(X_e) = 0 \Longleftrightarrow (X_e)\equiv 0,
\end{equation}
while \eqref{XomegaHolst} holds whenever $T_\gamma[\X{\omega}]\in \mathrm{ker}(\mathsf{W}^{(1,2)}_e)$.

Observe that $\omega\in\mathcal{A}_{P_{\partial M}}$ is a connection on the boundary, and $\mathsf{W}^{(1,2)}_e$ is a map on the tangent space $T_\omega\mathcal{A}_{P_{\partial M}}$. A vector field $\mathbb{X}_\omega = T_\gamma[\X{\omega}]\pard{}{\omega}$ in the kernel of $\widetilde{\varpi}_{PCH}$ acts on connections by adding an element of $\mathrm{ker}(\mathsf{W}^{(1,2)}_e)$, so that the result of flowing $\omega$ along $\mathbb{X}_\omega$ will be a connection $\omega'$ such that $\omega'-\omega \in \mathrm{ker}(\mathsf{W}^{(1,2)}_e)$. This defines an equivalence relation $\sim$ on connections on the boundary, and the reduction map $\pi_\sim\colon \mathcal{A}_{\iota^*P} \longrightarrow \mathcal{A}_{\iota^*P}^{red}\coloneqq\qsp{\mathcal{A}_{\iota^*P}}{\sim}\simeq\qsp{\mathcal{A}_{\iota^*P}}{\mathrm{ker}(\mathsf{W}^{(1,2)}_e)}$ sends $\omega$ to its equivalence class $[\omega]_e$. If we identify the coordinate in  $\mathcal{A}_{\iota^*P}$ with an equivalence class of connections, we get the explicit expression for the symplectic reduction map:
\begin{equation}
\pi\colon\begin{cases}
\mathbf{e}=e\\
\bom =[\omega]_e,
\end{cases}
\end{equation}
and pre-composing $\pi$ with the restriction map $\widetilde{\pi} \colon\mathcal{F}_{PCH} \longrightarrow \widetilde{\mathcal{F}}_{PCH}$ we get the surjective submersion to the space of boundary fields $\pi_{PCH}\colon \mathcal{F}_{PCH}\longrightarrow \mathcal{F}^{\partial}_{PCH}$. It is easy to check that $\widetilde{\alpha}_{cl}$ is horizontal w.r.t. $\mathrm{ker}(\mathsf{W}^{(1,2)}_e)$, and that the one-form
\begin{equation}
\alpha^{\partial}_{PCH}=-\frac12\int\limits_{\partial M}\hat{T}_\gamma\left[\bem\wedge\bem\wedge\delta\bom \right]
\end{equation}
is such that $\widetilde{\alpha}_{PCH}=\pi^*\alpha^{\partial}_{PCH}$.

Using Lemma \ref{Lemma:kerWe} again we can identify $\mathcal{A}_{\iota^*P}^{red}\simeq \Omega^2(\partial M,\Wedge{3}\mathcal{V}^\partial)$, after the choice of a reference connection. This choice induces a global Darboux chart and $\mathcal{F}^{\partial}_{PCH}$ is a cotangent bundle with Liouville form
\begin{equation}
\alpha^{\partial}_{PCH}=-\int\limits_{\partial M}\mathrm{Tr}\left[\mathbf{e}\wedge\delta(\mathbf{e}\wedge \hat{T}_\gamma[\bom]) \right]=\int\limits_{\partial M}\mathrm{Tr}\left[\mathbf{e}\wedge\delta\mathbf{t}_\gamma\right],
\end{equation}
under the identification $\mathbf{t}_\gamma:=-\bem\wedge T_\gamma[\bom]$.
\end{proof}

\begin{remark}
The choice of a complement of $\mathrm{ker}(\mathsf{W}^{(1,2)}_e)$ (cf. Remark \ref{dualsplitting}), induces the splitting $\omega=\tom + \mathrm{v}$ where $T_\gamma[\tom]$ is the component of the connection $T_\gamma[\omega]$ in the complement $\mathcal{C}^{(1,2)}$, and $T_\gamma[\mathrm{v}]\in\mathrm{ker}(\mathsf{W}^{(1,2)}_e)$. 
We write $\tom=\omega - \mathrm{v}$ so that we can read Eq. \eqref{XomegaHolst} as $(X_{\tom})=0$ and project to the symplectic reduction. If we denote by $(\mathbf{e},\bom)$ the coordinates on the space of boundary fields and by $\mathcal{A}_{\iota^*P}^{red}\coloneqq \qsp{\mathcal{A}_{\iota^*P}}{\mathrm{Ker}(\mathsf{W}^{(1,2)}_e)}$ the reduced space of connections on the boundary, the projection map reads
\begin{equation}
\pi\colon\begin{cases}
\mathbf{e}=e\\
\bom = \omega-\mathrm{v} = \tom.
\end{cases}
\end{equation}
\end{remark}

\begin{remark}
Note that the nondegeneracy of $e$ is necessary to ensure that the pre-boundary two-form be presymplectic (i.e., that its kernel be a subbundle). This is unlike the 3-dimensional case, where this is not necessary and Palatini--Cartan theory may be extended to degenerate co-frame fields. Furthermore, observe that in an open neighbourhood of $\bem$ one can choose $\bom$ in the complement of $\mathsf{W}_e^{(1,2)}$, and the variations are independent.
\end{remark}

\begin{lemma}\label{lemmakerintersect}
Consider the map $[\cdot,e]\colon \Omega^1(\partial M,\Wedge{2}\mathcal{V}^\partial)\longrightarrow \Omega^2(\partial M, \mathcal{V}^\partial)$, and the restriction of the metric $g^\partial\coloneqq g\big|_{\partial M} \equiv e^*\eta\big|_{\partial M}$. Let $\mathcal{K}\coloneqq \mathrm{ker}([\cdot,e])\cap\mathrm{ker}(\mathsf{W}_e^{(1,2)})$. We have that, for $\eta$ Euclidean or Minkowski,
\begin{equation}
\mathrm{dim}\left(\mathcal{K}\right)=2\mathrm{dim}(\mathrm{ker}g^\partial),
\end{equation}
Moreover, if $g^\partial$ is nondegenerate, the map $[\cdot, e]$ is surjective on $\Omega^2(\partial M, \mathcal{V}^\partial)$, and $\mathrm{ker}([\cdot , e])\subset\mathcal{C}^{(1,2)}$.
\end{lemma}

\begin{proof}
We restrict the map $[\cdot,e]$ to $\mathrm{ker}(\mathsf{W}_e^{(1,2)})$, and its kernel is then defined explicitly by
$$[v,e]^b_{ad}=v^{bc}_ag^\partial_{cd} - v^{bf}_dg^\partial_{fa}=0,$$
where we express $v$ in the basis $\{e_\mu\}$, and from Lemma \ref{Lemma:kerWe} we know that $v_a^{\mu\nu}e_\mu e_\nu\in\mathrm{ker}(\mathsf{W}_e^{(1,2)})$ iff $\sum_av_a^{ab}=0$ and $v=v_a^{bc}|\epsilon_{abc}|$. Now, we can choose coordinates in which $g^\partial_{ab}$ is diagonal, i.e. $g^\partial=\mathrm{diag}(\alpha_1,\alpha_2,\alpha_3)$, and the eigenvalues are only $0$ and $\pm 1$. If $\eta$ is a fibrewise Lorentzian metric on $M$, the only possible values are $(1,1,1), (1,1,-1),(1,1,0)$. If we allow for more general $\eta$'s we can also have cases $(1,-1,-1),(1,-1,0),(1,0,0)$. The result holds as well for $(0,0,0)$, trivially.

The equation in these coordinates is easily rewritten as
$$v_a^{bd}\alpha_d = v_d^{ba}\alpha_a,$$
for all $a,b,d$.
\paragraph{Case $g^\partial=\mathrm{diag}(1,1,\pm 1)$}
The equations reduce to
\begin{subequations}\begin{align}
v_1^{32} = v_2^{31};\\
\pm v_1^{23}=v_3^{21};\\
\pm v_2^{13}=v_3^{12};
\end{align}\end{subequations}
and
\begin{subequations}\begin{align}
v_1^{12}=v_2^{21}=0;\\
\pm v_1^{13}=v_3^{31}=0;\\
\pm v_2^{23}=v_3^{32}=0;
\end{align}\end{subequations}
that together with $\sum_a v_a^{ab}=0$ imply $v\equiv 0$.

With a similar argument, it is easy to show that, in the case $g^\partial$ nondegenerate, the map $[\cdot,e]$ is surjective on $\Omega^2(\partial M,\mathcal{V}^\partial)$: the kernel has 6 local dimensions and is obviously contained in $\mathcal{C}_{(1,2)}$, the complement of $\mathcal{W}_{(1,2)}$, since $[\cdot,e]\vert_{\mathcal{W}_{(1,2)}}$ is injective.

\paragraph{Case $g^\partial=\mathrm{diag}(1,\pm1,0)$}
This time we get the set of equations
\begin{subequations}\begin{align}
& v_1^{32} = \pm v_2^{31};\ \ \ v_3^{21}=0;\\
&v_3^{31}=v_2^{21}=v_3^{32}=v_1^{12}=0.
\end{align}\end{subequations}
From $\sum_a v_a^{ab}=0$ we get the additional constraint $v_1^{13}=-v_2^{23}$. $\mathcal{K}$ is then parametrised by $v_2^{b3}$ and is therefore 2-dimensional.

\paragraph{Case $g^\partial=\mathrm{diag}(1,0,0)$}
In this case the only equations that aren't automatically verified are 
$$v_2^{b1}=v_3^{b2}=0,$$
which imply $v_2^{21}=v_3^{31}=0$. Consequently, components $v_2^{b3},v_3^{b2}$ are free and parametrise the kernel, which is then 4-dimensional.
\end{proof}

\begin{corollary}\label{corollary}
If $g$ is Lorentzian and $g^\partial$ is nondegenerate, then we have a short exact sequence
\begin{equation}
0\longrightarrow\mathrm{ker}(\mathsf{W}_e^{(1,2)})\stackrel{[\cdot,e]}{\longrightarrow}\Omega^2(\partial M,\mathcal{V}^\partial)\stackrel{\mathsf{W}_e^{(2,1)}}{\longrightarrow} \Omega^3(\partial M,\Wedge{2}\mathcal{V}^\partial)\longrightarrow 0.
\end{equation}
\end{corollary}
\begin{proof}
Using Lemma \ref{Lemma:kerWe} we know that $\mathsf{W}_e^{(2,1)}$ is surjective and from Lemma \ref{lemmakerintersect} $[\cdot,e]$ is injective on $\mathrm{ker}(\mathsf{W}_e^{(1,2)})$. It is a matter of a simple computation to check explicitly that $e\wedge[v,e]=0$ for all $v\in\mathrm{ker}\mathsf{W}_e^{(1,2)}$. Alternatively one can define the map $\mathcal{E}\colon\Wedge{3}V\times V \longrightarrow \Wedge{2} V$ such that, for $w\in\Wedge{3}V$ and $a\in V$ 
$$\mathcal{E}(w,a)=\epsilon_{ijk}w^{ijk}u_iu_ja^l\eta(u_k,u_l)=\epsilon_{ijk}w^{ijk}u_iu_ja^l\eta_{kl}.$$ 
Then $e\wedge[v,e] = \mathcal{E}(e\wedge v,e) - v \wedge e^i\wedge e^j \eta_{ij}=0$. This implies that $\mathrm{im}([\cdot,e])\subset \mathrm{ker}(\mathsf{W}_e^{(2,1)})$, but since $\mathrm{dim}\mathcal{K}=0$ the reverse is also guaranteed and the sequence is exact.
\end{proof}

Recalling the definitions of the projections $p$ and $p'$ (cf. Remark \ref{projremark}):

\begin{corollary}\label{cor2}
If $g^\partial$ is nondegenerate, given a pair $(\omega,e)\in\mathcal{F}_{PCH}$, there exists a unique $\widetilde{v}\in\Omega^{1}(\partial M,\Wedge{2}\mathcal{V})$ that solves the system
\begin{equation}
\begin{cases}
e\wedge \widetilde{v} = 0 \\
[\widetilde{v},e] = p_{(2,1)}(d_\omega e),
\end{cases}
\end{equation}
where $p_{(2,1)}\colon \Omega^2(\partial M,\mathcal{V}^\partial)\longrightarrow \mathcal{W}_{(2,1)}$.
\end{corollary}
\begin{proof}
This is a consequence of the fact that the map 
\begin{equation}
\phi_e\equiv p_{(2,1)}\circ[\cdot, e]\vert_{\mathcal{W}_{(1,2)}} \colon \mathcal{W}_{(1,2)} \longrightarrow \mathcal{W}_{(2,1)}
\end{equation}
is an isomorphism.
\end{proof}

\begin{lemma}\label{lemmatilde}
Assume $g^\partial$ nondegenerate and consider the functions 
\begin{subequations}\begin{align}
\widetilde{v}\colon \widetilde{\mathcal{F}}_{PCH} \longrightarrow \mathrm{ker}(\mathsf{W}^{(1,2)}_e), &\ \ \ \ \widetilde{v}(\omega, e) = -\phi_e^{-1}(d_\omega e), \\
\widetilde{\omega}\colon \widetilde{\mathcal{F}}_{PCH} \longrightarrow \mathcal{A}_{\iota^*P}, &\ \ \ \ \widetilde{\omega}(\omega,e) = \omega + \widetilde{v}(\omega,e),
\end{align}\end{subequations}
where $\phi_e$ is as in Corollary \ref{cor2}. Then, $\tom$ is basic with respect to the reduction map $\pi_\sim\colon \mathcal{A}_{\iota^*P}\longrightarrow \mathcal{A}_{\iota^*P}^{red}$, and $[\tom]_e=[\omega]$. Moreover, 
$$pd_{\tom} e=0$$ 
and $\tom$ satisfies the equation
\begin{equation}\label{uniquerepequation}
d_{\tom} e = 0
\end{equation}
if and only if $e\wedge d_\omega e = 0$.
\end{lemma}

\begin{proof}
If $\mathbb{X}$ is a vertical vector field, i.e. $\iota_\mathbb{X}\widetilde{\varpi}_{PCH}=0$, one can easily verify that $\mathsf{L}_\mathbb{X}\widetilde{v}(\omega,e) = \mathsf{L}_\mathbb{X} \omega$. In fact, since $\mathsf{L}_{\mathbb{X}} e=0$ we have
$$\mathbb{X}(\phi^{-1}_e(pd_\omega e)) =\phi_e^{-1}(p[\mathbb{X}_\omega,e]) = \mathbb{X}_\omega = \mathsf{L}_{\mathbb{X}}\omega, $$
given that $\mathbb{X}_\omega\in\mathrm{ker}(\mathsf{W}_e^{(1,2)})$ (cf. Theorem \ref{Theo:PCHClassical}); thus $\mathsf{L}_\mathbb{X}\tom=0$, and $[\tom]_{e}=[\omega]_{e}$.

Moreover, by definition of $\tom$ we have $d_{\tom} e = d_\omega e + [\widetilde v,e]$ and, in virtue of Corollary \ref{cor2}, $[\widetilde{v}, e]=-p(d_\omega e)$, so that $d_{\tom}e= d_\omega e - pd_\omega e = p'd_\omega e$ and consequently $pd_{\tom}e=0$. Moreover, $d_{\tom}e$ vanishes if and only if $d_\omega e\in\mathrm{ker}(\mathsf{W}^{(2,1)}_e)$, that is when $e\wedge d_\omega e=0$.

\end{proof}

\begin{corollary}\label{corollarycheck}
Let $(\bem,\bom)\in\mathcal{F}^\partial_{PCH}$ and $g^\partial$ nondegenerate, then $\bem\wedge d_{\bom} \bem=0$ if and only if there exists, unique, a representative $\tom$ in the class $\bom\equiv[\omega]_e$ such that $d_{\tom}e=0$.
\end{corollary}

\begin{proof}
Follows immediately from Lemma \ref{lemmatilde}. 
\end{proof}

\begin{definition}\label{varphidef}
Denote by $\mathsf{S}\subset \widetilde{\mathcal{F}}_{PCH}$ the image of the map $\varphi\colon \widetilde{\mathcal{F}}_{PCH} \longrightarrow \widetilde{\mathcal{F}}_{PCH}$, defined by $\varphi(\omega,e) = (\tom(\omega,e),e)$.
\end{definition}

\begin{lemma}
We have
\begin{equation}
\mathsf{S}= \{(\tom,e) \in \widetilde{\mathcal{F}}_{PCH}\ \big|\ p(d_{\tom} e)=0 \}.
\end{equation}
\end{lemma}
\begin{proof}
We have already shown in Lemma \ref{lemmatilde} that every pair $(\tom, e)$ where $\tom=\tom(\omega,e)$ satisfies $p(d_{\tom}e)=0$. On the contrary, If $\tom$ satisfies the equation, then $\tom(\tom, e) = \tom$, showing that $(\tom,e)$ is in the image of $\varphi$.
\end{proof}

\begin{proposition}\label{Prop:structuralcharacterisation}
The map 
\begin{equation*}
\Phi\colon 
\begin{array}{ccc} 
\mathcal{F}^\partial_{PCH} & \longrightarrow & \mathsf{S}\\
 ([\omega]_e, e) &\longmapsto &  (\tom, e)
\end{array}\end{equation*}
is a symplectomorphism.
\end{proposition}
\begin{proof}
Choose a representative $\omega\in[\omega]_e$ and construct $\tom(\omega,e)$. This is independent of the representative, since changing $\omega$ by $\mathrm{v}\in\mathrm{Ker}(\mathsf(W)_e^{(1,2)})$ leaves $\tom(\omega,e)$ invariant, as shown in Lemma \ref{lemmatilde}. On the other hand, choosing $(\tom, e)$ in $\mathsf{S}$ uniquely determines the equivalence class $[\tom]_e$ as the fibre of $\varphi$ above $(\tom,e)$ coincides with the fibre of $\pi\colon \widetilde{F}_{PCH} \longrightarrow \mathcal{F}^\partial_{PCH}$.

Computing the symplectic form $\varpi^\partial_{PCH}$ on $(\tom,e)$ yields
$$\intl_{\partial M} \hat{T}_\gamma[e \delta e \delta (\omega + \widetilde{v}(\omega,e)] = \intl_{\partial M} \hat{T}_\gamma[e \delta e \delta \omega + e\delta e \delta \widetilde{v}]$$
for some representative $\omega \in [\omega]_e$. However, since $e\wedge \widetilde{v}=0$ we have $e\delta \check{v} = \delta (e\wedge \widetilde{v}) - \delta e \wedge \widetilde{v}=- \delta e \wedge \widetilde{v}$ and we compute
$$\intl_{\partial M} \hat{T}_\gamma[e\delta e \delta \widetilde{v}] = -\intl_{\partial M} \hat{T}_\gamma[\delta e \delta e \widetilde{v}] = 0.$$
\end{proof}

\subsection{Constraints in PCH theory}\label{PCHCA}
In this section we will consider the structure of the constraints defined by equations \eqref{classconstraints} in the bulk. Compare this with the general discussion in section \ref{Sect:Constraintstheory}. 

We consider the Euler--Lagrange equations in the space of bulk fields $\mathcal{F}_{PCH}$ and define the \emph{pre-constraints} to be given by their restriction to the boundary. We will show how one can construct local functionals that are basic with respect to the symplectic fibration $\pi\colon\widetilde{\mathcal{F}}_{PCH}\longrightarrow \mathcal{F}^{\partial}_{PCH}$.

\begin{remark}[On the definition of constraints]\label{ELproof}
In the bulk, we have that the Euler--Lagrange equation $e\wedge d_\omega e=0$ is equivalent to $d_\omega e=0$. We can regard the restricted equation $d_\omega e\big|_{\partial M}=0$ as defining the constraint, together with $e\wedge T_\gamma[F_\omega] \big|_{\partial M}$. In virtue of Lemma \ref{lemmatilde}, requiring $d_\omega e=0$ on the boundary is equivalent to the condition $ed_{\omega'}e$ for every other $\omega'\in[\omega]_e$, and $\tom(\omega',e)\equiv\omega$ is the unique representative satisfying \eqref{uniquerepequation}. Alternatively, we could pick $(e\wedge d_\omega e)\vert_{\partial M}=0$ as constraint, and know that there exists, unique, a representative $\tom$ in the class of $\omega$, such that $d_{\tom}e=0$.
\end{remark}

\begin{remark}\label{constraintequivalence}
Another way to observe this phenomenon is the following: consider the function
$$\intl_{\partial M}\mathrm{Tr}[\beta\wedge d_{\tom} e], $$
with $\beta\in\Omega^1(\partial M,\Wedge{3}\mathcal{V})$. However, we know that $\pi(d_{\tom}e) = 0$ and $d_{\tom}e\equiv \pi' d_{\tom}e$. Since $\mathcal{C}_{(i,j)}^*\simeq \mathrm{Im}(\mathsf{W}_e^{(2-i,3-j)})$ there exists $\alpha\in\Omega^0(\partial M, \Wedge{2}\mathcal{W})$ such that $\beta = e\wedge \alpha$ and
$$\intl_{\partial M}\mathrm{Tr}[\beta \wedge d_{\tom}e] = \intl_{\partial M} \mathrm{Tr}[\alpha \wedge e\wedge d_{\tom} e].$$
\end{remark}

\begin{remark}
We can split the constraint $d_{\tom} e = 0$ into its two projections, i.e. into the \emph{structural constraint} $p d_{\tom} e = 0$ and the \emph{residual constraint} $e\wedge d_{\tom} e = 0$. Proposition \ref{Prop:structuralcharacterisation} shows that, under the nondegeneracy assumption on $g^\partial$, imposing the structural constraint is equivalent to the reduction with respect to the kernel of the presymplectic form. We will prove in a moment that the residual constraint then defines a coisotropic submanifold (and that additionally imposing the other constraints $e T_\gamma[F_{\tom}]=0$ also defines a coisotropic submanifold).
\end{remark}

\begin{remark}\label{remarkkerneldegrees}
From Lemma \ref{lemmakerintersect}, in the degenerate case $\mathcal{K}$ has dimension 2, which means that to project $d_{\omega}e$ to the space of boundary fields we must impose two extra conditions to eliminate the residual degrees of freedom in $\mathcal{K}$. From now on we will assume that $g^\partial$ is nondegenerate, which is equivalent to requiring that the signature of the induced metric on the boundary is either space-like or time-like.
\end{remark}

\begin{remark}
Observe that from Lemma \ref{lemmatilde} we know that $\tom$ is basic, i.e. $\mathsf{L}_{\mathbb{X}}\tom=\mathsf{L}_{\mathbb{X}} e=0$; from this we immediately conclude that $\mathsf{L}_{\mathbb{X}} (e\wedge d_{\tom}e)=\mathsf{L}_{\mathbb{X}}(e\wedge F_{\tom})=0$ for $\mathbb{X}\in\Gamma(\mathrm{ker}(\widetilde{\varpi}_{PCH}))$, showing that the constraints are basic in $\widetilde{\mathcal{F}}_{PCH}$ with respect to pre-symplectic reduction.
\end{remark}

\begin{theorem}\label{constraintstheorem}
The functions on $\mathsf{S}$:
\begin{subequations}\begin{align}
\widetilde{L}_\alpha=&\intl_{\partial M} \hat{T}_\gamma[\alpha \wedge e \wedge d_{\tom}{e}], \\
\widetilde{J}_\mu =&\intl_{\partial M}\hat{T}_\gamma\left[\mu\wedge {e}\wedge F_{\tom}\right] + \mathrm{Tr}\left[\Lambda \mu\wedge e^3\right],\end{align}
\end{subequations}
with $\alpha\in\Omega^0(M,\bigwedge^2\mathcal{V}^\partial),\ \mu\in\Omega^0(M,\iota^*\mathcal{V})$ and $\tom$ as in Lemma \ref{lemmatilde}, define a coisotropic submanifold $\mathcal{C}_{PCH}\subset\mathsf{S}\simeq\mathcal{F}^\partial_{PCH}$ with respect to the symplectic form 
$$\varpi^\partial=-\intl_{\partial M} \hat{T}_\gamma\left[e\delta e \delta \tom\right]$$
with the following algebraic relations:
\begin{subequations}\label{algstructure}\begin{align}
\{\widetilde{L}_\alpha,\widetilde{L}_{\alpha'}\} & = \widetilde{L}_{[\alpha',\alpha]}\\
\{\widetilde{L}_\alpha,\widetilde{J}_\mu\} & = \widetilde{J}_{[\alpha,\mu]} + \widetilde{L}_{H^{\alpha,\mu}} \\
\{\widetilde{J}_\mu,\widetilde{J}_{\mu'}\} & = \widetilde{L}_{X^{\mu\mu'}}
\end{align}\end{subequations}
where $H^{\alpha,\mu}, X^{\mu\mu'}$ are local functions of the fields and the Lagrange multipliers whose explicit expressions, presented in Equations \eqref{Hmualpha} and \eqref{Ymumu}, are not relevant for the statement.
\end{theorem}

\begin{proof}
In order to compute the Hamiltonian vector fields of the functions defining the constraints, we have to consider that the variation of $\tom$ is constrained. As a matter of fact, from the defining equation $p d_{\tom}e=0$, splitting the variation in $\delta\tom = p\delta\tom + p'\delta\tom$ and recalling that the projection to the kernel $p$ depends on $e$, we get
$$
[p\delta\tom, e] = -(\delta_ep)d_{\tom}e + pd_{\tom}e\delta e - p[p'\delta\tom, e], 
$$
so that
\begin{equation}
p\delta\tom =: \mathrm{A}(\delta e) + \mathrm{B}(p'\delta\tom),
\end{equation}
where we defined $\mathrm{A}\colon \Omega^{1}(\partial M,\mathcal{W}) \longrightarrow \mathcal{W}_{(1,2)}$ and $\mathrm{B}\colon \mathcal{C}_{(1,2)} \longrightarrow \mathcal{W}_{(1,2)}$, using the notation of Remark \ref{dualsplitting}. In particular, we have the explicit expression for $c\in\mathcal{C}_{(1,2)}$
$$B(c)=-\phi_e^{-1}(p[c,e]).$$
Observe that $\delta e$ and $p'\delta\tom$ are free variations, and that the map $A$ has odd parity. In a similar fashion, a generic Hamiltonian vector field $X$ will have to satisfy $pX_{\tom} = A(X_e) + B(p' X_{\tom})$.

Because of the structure of $\widetilde{L}_\alpha$, only the free variations $\delta e$ and $p'\delta\tom$ appear in $\delta\widetilde{L}_\alpha$: 
$$
\delta\widetilde{L}_\alpha=- \intl_{\partial M} \hat{T}_\gamma\left[\delta e \wedge e \wedge d_{\tom}\alpha  + p'\delta\tom\wedge e\wedge [\alpha,e]\right].
$$
Hence, its Hamiltonian vector field $\mathbb{L}_\alpha$ such that $\iota_{\mathbb{L}_\alpha}\varpi^\partial=\delta \widetilde{L}_\alpha$, simply needs to satisfy
\begin{align*}
e\wedge(\mathbb{L}^\alpha)_{e}&= e\wedge [\alpha, e];\\
e\wedge (\mathbb{L}^\alpha)_{\tom}&= -e\wedge d_{\tom}\alpha.
\end{align*}
Turning to $\widetilde{J}_\mu$ and recalling that $d_{\tom}e=p'd_{\tom}e$ we have 
\begin{multline}
\delta\widetilde{J}_\mu = \intl_{\partial M}\hat{T}_\gamma\left[ \delta e\wedge \mu \wedge\left(F_{\tom} + 3\Lambda e^2\right) + d_{\tom}(\mu \wedge e)\wedge \delta \tom\right]\\
 = \intl_{\partial M}\hat{T}_\gamma\left[ \delta e\wedge \mu\wedge \left(F_{\tom} + 3\Lambda e^2\right) + d_{\tom}\mu\wedge e\wedge p' \delta \tom - \mu\wedge d_{\tom}e\wedge \delta\tom \right]\\
=\intl_{\partial M}\hat{T}_\gamma\left[ \delta e \mu \left(F_{\tom} + 3\Lambda e^2\right) + d_{\tom}\mu e p' \delta \tom - \mu p'd_{\tom}e p'\delta\tom - \mu p' d_{\tom}e \left(\mathrm{A}(\delta e) +\mathrm{B}(p'\delta\tom)\right)\right]\\
=\intl_{\partial M}\hat{T}_\gamma\big[ \delta e\wedge\left[ \mu\wedge \left(F_{\tom} + 3\Lambda e^2\right) + \mathrm{A}^\dag\left(\mu\wedge p' d_{\tom}e\right)\right] \\
+ \left[d_{\tom}\mu\wedge e - p^\dag(\mu\wedge p'd_{\tom}e) - \mathrm{B}^\dag\left(\mu\wedge p' d_{\tom}e\right)\right]\wedge p'\delta\tom\big],
\end{multline}
where we introduced the adjoint maps $\mathrm{A}^\dag,\mathrm{B}^\dag$. Thus, we have
\begin{subequations}\begin{align}
\label{hamJ1}
e\wedge (\mathbb{J}^\mu)_{e} &= -d_{\tom}\mu\wedge e + (p^\dag + \mathrm{B}^\dag)\left(\mu\wedge p' d_{\tom}e\right) \\ \label{hamJ2}
e\wedge p'(\mathbb{J}^\mu)_{\tom} &= \mu\wedge \left(F_{\tom} + 3\Lambda e  \wedge e\right) + \mathrm{A}^\dag\left(\mu \wedge p' d_{\tom}e\right).
\end{align}\end{subequations}
Observe that imposing $d_{\tom}e=0$ would be sufficient to show that the ideal is coisotropic, but we are interested in explicitly writing the algebraic relations between constraints.

Let us compute:
\begin{align*}
\{\widetilde{L}_\alpha, \widetilde{L}_{\alpha'}\} &= \mathbb{L}^\alpha \widetilde{L}_{\alpha'} \\
& =\intl_{\partial M} \hat{T}_\gamma\left[\alpha' \wedge d_{\tom} \left(e\wedge (\mathbb{L}^\alpha)_{e}\right) + \frac12\alpha'\wedge [(\mathbb{L}^\alpha)_{\tom},e\wedge e] \right]\\
& =-\frac12\intl_{\partial M} \hat{T}_\gamma\left[d_{\tom}\alpha'\wedge[\alpha,e\wedge e] + \alpha'\wedge[d_{\tom}\alpha,e\wedge e]\right]\\
& = \intl_{\partial M}\hat{T}_\gamma\left[[\alpha',\alpha]\wedge e\wedge d_{\tom}e \right]= \widetilde{L}_{[\alpha',\alpha]},
\end{align*}
and
\begin{align*}
\{\widetilde{J}_\mu,\widetilde{L}_\alpha\} & =  \mathbb{J}^\mu \widetilde{L}_\alpha \\
& = -\intl_{\partial M} \hat{T}_\gamma\left[e \wedge (\mathbb{J}^\mu)_{e}\wedge d_{\tom} \alpha +  (\mathbb{J}^\mu)_{\tom}\wedge e \wedge [\alpha,e] \right]\\
& = \intl_{\partial M}\hat{T}_\gamma\Big[ d_{\tom}\mu\wedge e\wedge d_{\tom}\alpha - (p^\dag + \mathrm{B}^\dag)(\mu d_{\tom}e) \wedge d_{\tom}\alpha \\
&\phantom{\intl_{\partial M}\hat{T}_\gamma\Big[} - \mu\wedge \left(F_{\tom} + 3\Lambda e\wedge e\right)\wedge [\alpha,e]  - \mathrm{A}^\dag(\mu d_{\tom} e)\wedge [\alpha,e]\Big]\\
& =  \intl_{\partial M} \hat{T}_\gamma\Big[\mu\wedge d_{\tom}e\wedge d_{\tom}\alpha + [F_{\tom},\alpha] \wedge\mu\wedge e - \mu\wedge d_{\tom}e\wedge (\mathrm{id} + \mathrm{B})p'd_{\tom}\alpha\\ 
&\phantom{\intl_{\partial M}\hat{T}_\gamma\Big[}- \mu\wedge \left(F_{\tom} + 3\Lambda e\wedge e\right)\wedge [\alpha,e] - \mu\wedge d_{\tom}e \wedge \mathrm{A}([\alpha,e])\Big]\\
& = \intl_{\partial M}\hat{T}_\gamma\Big[p' d_{\tom} ep^\dag\left[\left((\mathrm{id} - p' - \mathrm{B}\circ p')d_{\tom}\alpha - \mathrm{A}([\alpha,e])\right)\wedge \mu\right] \\
&\phantom{\intl_{\partial M}\hat{T}_\gamma\Big[} + [\alpha,\mu]\wedge e \wedge \left(F_{\tom} + \Lambda e^2\right)\Big]\\
& = \intl_{\partial M} \hat{T}_\gamma\left[[\alpha, \mu]\wedge e\wedge F_{\tom} \right] + \intl_{\partial M}\hat{T}_\gamma\left[e \wedge d_{\tom} e \wedge H^{\alpha,\mu}\right]= \widetilde{J}_{[\alpha,\mu]} + \widetilde{L}_{H^{\alpha,\mu}},
\end{align*}
where we used that $\hat{T}_\gamma[3\mu[\alpha,e]ee] = \mathrm{Tr}[\mu[\alpha,e^3]] = - \mathrm{Tr}[[\alpha,\mu],e^3]$ and that through $\mathsf{W}_e^{-1}$, defined on the image of $p^\dag$ and $B^\dag$, we can define $H^{\alpha,\mu}\in\Omega^0(\partial M,\Wedge{2}\mathcal{V})$ such that 
\begin{equation}\label{Hmualpha}
H^{\alpha,\mu}\coloneqq \mathsf{W}_e^{-1}\left[\left((p - \mathrm{B}\circ p')d_{\tom}\alpha - \mathrm{A}([\alpha,e])\right)\wedge \mu\right].
\end{equation}

There is also another way of expressing $H^{\alpha,\mu}$. The Hamiltonian vector field $\mathbb{L}^{\alpha}$ must preserve the constraint $p (d_{\tom} e)=0$, so that, imposing $\mathsf{L}_{\mathbb{L}^\alpha}(p(d_{\tom}e))=0$ and recalling that the projection $p$ depends on $e$, one obtains
\begin{multline}\label{liederivativeofconstraint}
\mathsf{L}_{\mathbb{L}^\alpha}(pd_{\tom}e)=0\ \iff\ (\mathsf{L}_{\mathbb{L}^\alpha}p)(d_{\tom} e) + p[(\mathbb{L}^\alpha)_{\tom},e] +  pd_{\tom}[\alpha,e]=0 \\
\iff  p[p'(\mathbb{L}^\alpha)_{\tom},e] + p[p(\mathbb{L}^\alpha)_{\tom},e] + [p d_{\tom}\alpha,e] + p[p'd_{\tom}\alpha,e] + p[\alpha,d_{\tom}e] + (\mathsf{L}_{\mathbb{L}^\alpha}p)(d_{\tom} e) =0\\
\iff\ \left[p(\mathbb{L}^\alpha)_{\tom} + pd_{\tom}\alpha, e\right] = - (\mathsf{L}_{\mathbb{L}^\alpha}p)(d_{\tom} e) - p[\alpha,d_{\tom}e].
\end{multline}
Define now $\psi_{\alpha}\in\mathcal{W}_{(2,1)}$ to be $\psi_{\alpha}\coloneqq p(\mathbb{L}^\alpha)_{\tom} + pd_{\tom}\alpha$; as a consequence of \eqref{liederivativeofconstraint} we have
\begin{equation}
\psi_{\alpha}= - \phi_e^{-1}\left[(\mathsf{L}_{\mathbb{L}^\alpha}p)(d_{\tom} e) + p[\alpha,d_{\tom}e]\right],
\end{equation}
and notice that it vanishes on-shell (i.e. when $d_{\tom}e=0$). If we compute $\mathbb{L}^\alpha(\widetilde{J}_\mu)$, which on the one hand coincides with $\{\widetilde{J}_\mu,\widetilde{L}_\alpha\}$, we get
\begin{multline}
\mathbb{L}^\alpha(\widetilde{J}_\mu)= \intl_{\partial M} \hat{T}_\gamma\Big[ (\mathbb{L}^\alpha)_{e}\wedge  \mu\wedge (F_{\tom} + 3\Lambda e^2) + p'(\mathbb{L}^\alpha)_{\tom} \wedge d_{\tom}(\mu\wedge e) - p(\mathbb{L}^\alpha)_{\tom}\wedge \mu \wedge d_{\tom}e\Big]\\
=\intl_{\partial M} \hat{T}_\gamma\Big[ [\alpha,e] \wedge \mu \wedge (F_{\tom} + 3\Lambda e^2) - d_{\tom}\alpha\wedge d_{\tom}(\mu \wedge e) - \left(p(d_{\tom} \alpha) + p (\mathbb{L}^\alpha)_{\tom}\right) \wedge \mu\wedge d_{\tom}e\Big]\\
=\intl_{\partial M} \hat{T}_\gamma\Big[ [\alpha,\mu] \wedge e \wedge (F_{\tom} + \Lambda e^2) - p^\dag(\psi_\alpha \wedge \mu) p' (d_{\tom}e),
\end{multline}
from which we conclude
\begin{equation}
e\wedge H^{\alpha,\mu} = - p^\dag(\psi_{\alpha}\wedge \mu)
\end{equation}
so that $\mathsf{W}_e^{-1}(p^{\dag}(\psi_{\alpha}\wedge \mu))=H^{\alpha,\mu}$ and $e\wedge H^{\alpha,\mu}$ vanishes on shell.

Finally, we compute (we omit wedge symbols from the third line on)
\begin{multline}
\{\widetilde{J}_{\mu'},\widetilde{J}_{\mu}\}= \mathbb{J}^{\mu'}\widetilde{J}_\mu = \intl_{\partial M} \hat{T}_\gamma\Big[ (\mathbb{J}^{\mu'})_{e} \wedge\left[ \mu\wedge \left(F_{\tom} + 3\Lambda e^2\right) + \mathrm{A}^\dag(\mu\wedge p'(d_{\tom} e))\right] \\
 + p'(\mathbb{J}^{\mu'})_{\tom}\wedge\left[ d_{\tom}(e\wedge \mu) - (p^\dag + \mathrm{B}^\dag)(\mu\wedge p'(d_{\tom}e))\right] \Big]\\
=\intl_{\partial M}\hat{T}_\gamma\Big[-d_{\tom}\mu'\left[\mu(F_{\tom} + 3\Lambda e^2) + \mathrm{A}^\dag(\mu d_{\tom} e)\right] +\mathrm{A}^\dag(\mu' d_{\tom}e) \mathsf{W}_e^{-1}\left[(p^\dag + \mathrm{B}^\dag)(\mu d_{\tom} e)\right]\\
+\mu'(F_{\tom} + 3\Lambda e^2)d_{\tom}\mu + \mathrm{A}^\dag(\mu' d_{\tom}e)d_{\tom}\mu - \mu'(F_{\tom} + 3\Lambda e^2)\mathsf{W}_e^{-1}\left[(p^\dag + \mathrm{B}^\dag)(\mu d_{\tom} e)\right] \\
+ \mathsf{W}_e^{-1}\left[(p^\dag + \mathrm{B}^\dag)(\mu' d_{\tom}e)\right]\left[\mu(F_{\tom} + 3\Lambda e^2) + \mathrm{A}^\dag(\mu d_{\tom}e)\right] \Big]\\
=\intl_{\partial M}\hat{T}_\gamma\Big[ -\left[d_{\tom}\mu' \mu - \mu'd_{\tom}\mu\right] \left[F_{\tom} + 3\Lambda e^2\right] - \left[\mathrm{A}(d_{\tom}\mu') \mu - \mathrm{A}(d_{\tom}\mu)\mu'\right]p'd_{\tom} e \\
 + \left[\mathsf{W}_e^{-1}\left[(p^\dag + \mathrm{B}^\dag)(\mu' p'd_{\tom}e)\right]\mu - \mathsf{W}_e^{-1}\left[(p^\dag + \mathrm{B}^\dag)(\mu p'd_{\tom}e)\right]\mu'\right][F_{\tom} + 3\Lambda e^2]\\ 
 +\left[\mathrm{A}\left[\mathsf{W}_e^{-1}\left[(p^\dag + \mathrm{B}^\dag)(\mu'p'd_{\tom} e)\right]\right] \mu - \mathrm{A}\left[\mathsf{W}_e^{-1}\left[(p^\dag + \mathrm{B}^\dag)(\mu p'd_{\tom} e)\right]\right] \mu'\right]p' d_{\tom} e\Big]
 \end{multline}

We proceed by defining\footnote{We can do this because $\mathsf{W}_e^{(1,2)}$ is surjective.} $Z^\mu\in\Omega^1(\partial M,\Wedge{2}\mathcal{V})$ such that $p(Z^\mu)=0$ and $e\wedge Z^\mu=\mu F_{\tom}$, so that 
$$\mathsf{W}_e^{-1}\left((p^\dag + \mathrm{B}^\dag)(\mu' p'd_{\tom}e)\right)\mu F_{\tom} = (p^\dag + \mathrm{B}^\dag)(\mu' p'd_{\tom}e)Z^\mu=\mu' p'd_{\tom}e(p' + \mathrm{B})Z^\mu.$$
Thus, we can write:
 \begin{multline}
\{\widetilde{J}_{\mu'},\widetilde{J}_{\mu}\}=-\intl_{\partial M}\hat{T}_\gamma\Big[ d_{\tom}(\mu'\mu)(F_{\tom} + 3\Lambda e^2)\Big]\\
- \intl_{\partial M}\hat{T}_\gamma\Big[p^\dag\left[\mathrm{A}(d_{\tom}\mu')\wedge\mu + \mu'\wedge (p' + \mathrm{B})(Z^\mu + \Lambda \mu\wedge e)\right] p'(d_{\tom} e) \\
- p^\dag\left[\mathrm{A}\left[\mathsf{W}_e^{-1}\left((p^\dag + \mathrm{B}^\dag)(\mu'\wedge p'(d_{\tom} e))\right)\right]\wedge \mu\right]p'(d_{\tom} e) - \{\mu\leftrightarrow\mu'\}\Big]\\
= 6\Lambda\intl_{\partial M}\hat{T}_\gamma\left[ \mu' \wedge \mu \wedge e\wedge d_{\tom}e  \right]  + \intl_{\partial M}\hat{T}_\gamma\left[Y^{\mu'\mu} \wedge e\wedge d_{\tom}e  \right] \\
= \intl_{\partial M} \hat{T}_\gamma\left[X^{\mu'\mu} \wedge e \wedge d_{\tom} e\right] = \widetilde{L}_{X^{\mu'\mu}},
\end{multline}
where we used the Bianchi identity and defined 
\begin{multline}\label{Ymumu}
Y^{\mu'\mu}\coloneqq - \mathsf{W}_e^{-1}p^\dag\Big[ \mathrm{A}(d_{\tom}\mu')\wedge\mu + \mathrm{A}\left[\mathsf{W}_e^{-1}\left((p^\dag + \mathrm{B}^\dag)(\mu'p'd_{\tom} e)\right)\right]\wedge \mu \\
+ \mu'\wedge (p' + \mathrm{B})(Z^\mu + \Lambda \mu\wedge e) - \{\mu\leftrightarrow \mu'\}\Big],\end{multline}
so that $X^{\mu'\mu}=6\Lambda \mu'\wedge\mu + Y^{\mu'\mu}$, because of Lemma \ref{lemmatilde} together with the observation in Remark \ref{constraintequivalence}.
\end{proof}

\begin{corollary}\label{constraintscorollary}
In the symplectic manifold 
$$\mathcal{F}_{PCH}^{\partial}\longrightarrow \Omega^1_{nd}(\partial M,\mathcal{V}^\partial)$$ 
with symplectic form $\varpi^\partial_{PCH}$ as in Eq. \eqref{classicalboundaryform}, the vanishing locus $\mathcal C_{PCH}$ of the functions:
\begin{equation}
\mathbf{L}_\alpha=\intl_{\partial M} \hat{T}_\gamma[\alpha\wedge \mathbf{e}\wedge d_{\bom}\mathbf{e}];\ \ \mathbf{J}_\mu = \intl_{\partial M}\hat{T}_\gamma\left[ \mu \wedge\bem\wedge F_{\bom}\right] + \mathrm{Tr}\left[\Lambda\mu\wedge\bem^3\right],
\end{equation}
with $\mu\in\Omega^0(\partial M,\mathcal{V}^\partial)$ and $\alpha\in\Omega^0(M,\bigwedge^2\mathcal{V}^\partial)$ is coisotropic. We have the algebraic structure:
\begin{subequations}\begin{align}
\{\mathbf{L}_\alpha, \mathbf{L}_{\alpha'}\} & =  \mathbf{L}_{[\alpha',\alpha]}\\
\{\mathbf{J}_\mu,\mathbf{J}_{\mu'}\} & = \mathbf{L}_{X^{\mu\mu'}}\\
\{\mathbf{L}_\alpha,\mathbf{J}_\mu\} & = \mathbf{J}_{[\alpha,\mu]} + \mathbf{L}_{H^{\alpha,\mu}}
\end{align}\end{subequations}
where $X^{\mu\mu}$ and $H^{\alpha,\mu}$ are defined as in Theorem \ref{constraintstheorem}.
\end{corollary}

\begin{proof}
This follows automatically from Theorem \ref{constraintstheorem} and the projectability of the functions $\{\widetilde{J}_\alpha,\widetilde{J}_\lambda\}$ along $\pi\colon \widetilde{F}_{PCH}\longrightarrow\mathcal{F}^\partial_{PCH}$, after writing $\widetilde{L}_\alpha = \pi^*\mathbf{L}_\alpha$ and $\widetilde{J}_\mu = \pi^*\mathbf{J}_\mu$.
\end{proof}

\begin{remark}
The constraint algebra found in Theorem \ref{constraintstheorem} can be linked to the action of 4 dimensional diffeomorphism with the following observation. Since $\mu\in\Omega^0(\partial M, \mathcal{V}^\partial)$ it can be written (at least locally) as $\mu= \xi^n e_n + \iota_\xi \bem$, where we choose $e_n$ to be in the complement of $\mathrm{im}(\bem)$ in V.

We write $\mathbf{J}_\xi\equiv \mathbf{J}_{\iota_\xi\bem}$ and we can compute the Hamiltonian vector field in this case. Observing that $\Lambda\mu e^3\not=0$ only if $\mu=\xi^ne_n$ we get
$$\bem\wedge(\mathbb{J}_\xi)_{\bem}=\frac12 d_{\bom}\iota_\xi(\bem \wedge \bem);\ \  (\mathbb{J}_\xi)_{\bom} = -\iota_\xi F_{\bom}.$$
On the constraint surface, i.e. imposing $\bem d_{\bom}\bem$, we can conclude that $(\mathbb{J}_\xi)_{\bem} = -L_\xi^{\bom}\bem$ and $(\mathbb{J}_\xi)_{\bom} = -\iota_\xi F_{\bom}$, which reproduces the (covariant) action of vector fields on a connection and a co-frame as described in \cite{CS2,THESIS,CSS}.

If the principal bundle over $\partial M$ were trivial, or alternatively if we worked with the simply-connected cover of $SO(3,1)$, we could regard $\bom$ as a global one-form on the boundary, and make sense of $\iota_\xi\bom$. This would enable us to look at $\mathbf{D}_\xi\coloneqq\mathbf{J}_\xi + \mathbf{L}_{\iota_\xi\bom}$, which has the Hamiltonian vector field
\begin{equation}
\mathbb{D}_\xi= L_\xi\bem \pard{}{\bem} + L_\xi\bom \pard{}{\bom}.
\end{equation}
However, the bracket of two constraints $\mathbf{J}_\xi$ can be found with a straightforward computation to be (we drop $T_\gamma$ for simplicity)

\begin{align*}
\{\mathbf{J}_\xi, \mathbf{J}_{\xi}\} =\mathbb{J}_{\xi}\mathbf{J}_{\xi}&=\frac12\intl_{\partial M} \mathrm{Tr}\left[\iota_{\xi}\left( (\mathbb{J}_\xi)_{\bem}\wedge \bem\right) F_{\bom} + \iota_{\xi}(\bem\wedge \bem) d_{\bom} ( \mathbb{J}_\xi)_{\bom} \right]\\ 
&= \frac12 \intl_{\partial M} \mathrm{Tr}\left[ \iota_{\xi} d_{\bom} \iota_\xi (\bem\wedge \bem) F_{\bom} + \iota_\xi d_{\bom} \iota_{\xi} (\bem\wedge\bem) F_{\bom}\right]\\
&= \intl_{\partial M}\mathrm{Tr}\left[ \iota_{[\xi,\xi]}\bem \bem F_{\bom} - \frac12 d_{\bom} \iota_\xi \iota_{\xi}(\bem\wedge\bem) F_{\bom} + \frac12\iota_\xi\iota_\xi d_{\bom}(\bem\wedge\bem) F_{\bom}\right]\\
&= \mathbf{J}_{[\xi,\xi]} + \mathbf{L}_{\iota_\xi\iota_\xi F_{\bom}}.
\end{align*}
When $\mu=\xi^n e_n$, instead, the computation is given in Theorem \ref{constraintstheorem}. \end{remark}

Assuming that $g^\partial$ is nondegenerate, we can establish a relationship between the reduced phase space of PCH theory and of EH theory with the following.

\begin{theorem}\label{PCHtoEH}
The reduced phase space for Palatini--Cartan--Holst theory maps to that of Einstein--Hilbert, namely, denoting by $\mathcal{C}_{EH}\subset\mathcal{F}^\partial_{EH}$ the submanifold of canonical constraints for EH theory, and by $\underline{ \mathcal{G}}$ the coisotropic reduction of the zero locus of the function $\mathbf{L}$, we have the symplectomorphism 
\begin{equation}
\varphi\colon \underline{\mathcal{G}}\longrightarrow \mathcal{F}^\partial_{EH},
\end{equation}
and $\mathcal{C}_{EH}=\varphi\left( \mathcal{C}_{PCH}\right)$.
\end{theorem}

\begin{proof}
The strategy is as follows: we will first give a convenient description of the critical locus $\mathcal{G}=\mathbf{L}^{-1}(0)$ and then reduce it using the Marsden--Weinstein construction \cite{MW}, understanding $\mathbf{L}$ as the moment map for an $\mathfrak{so}(3,1)$ action on $\mathcal{F}^\partial_{PCH}$. We show that the reduced locus $\underline{\mathcal{G}}$ is symplectomorphic to the space of boundary fields of Einstein--Hilbert theory, $\mathcal{F}^\partial_{EH}$, and finally that the critical locus of the residual constraint inside $\underline{\mathcal{G}}$ maps to the reduced critical locus $\mathcal{C}_{EH}\subset\mathcal{F}^\partial_{EH}$.

Consider the induced tetrad on the boundary $\bem\colon T\partial M\longrightarrow \mathcal{V}$ and observe that the subspace $W=\mathrm{Span}\{\bem_a\}_{a=1\dots 3}$ is a space-like (resp. time-like) subspace of $V$, with respect to $\eta$, since $g^\partial=e^*\eta\vert_{\partial M}$ is nondegenerate. We can then use an algorithm to find an $\eta$-orthonormal basis of W - say $\{w_{\underline{i}}\}$, and the change of basis matrix $\tem_a^{\underline{i}}$ is invertible and defines a triad $\tem\colon T_x\partial M \longrightarrow W$. Moreover, we can complete $\{w_{\underline{i}}\}$ to a basis of V by a vector $\{w_0\}$ with norm $\eta_{00}=\pm 1$, depending on the signature of $\eta\vert_{W}$. 

Consider the function\footnote{Notice that this constraint coincides with the one previously considered upon picking $c=T_\gamma[\alpha]$.} 
$$\mathbf{L}_c\equiv\intl_{\partial M} \mathrm{Tr}[c\wedge  \bem \wedge d_{\bom} \bem].$$
This enforces the constraint $\bem\wedge d_{\bom}\bem = 0$. Let us write $\bom = \bg_{\tem} + \bA$, where $\bg_{\tem}$ is the connection compatible with the triad $\tem_{a}^{\underline{i}}$ such that $d_{\bg_{\tem}}\tem_a^{\underline{i}}=0$.

In the new basis $\{w_0,w_{\underline{i}}\}$ we have the splitting of the constraint
\begin{equation}\label{splitconstraints}
\begin{cases}
[\bem \wedge d_{\bom} \bem]^{\underline{i}\underline{j}} = 0\\
[\bem \wedge d_{\bom} \bem]^{0\underline{j}} = 0
\end{cases}
\end{equation}
and we can use $c\in\Omega^0(\partial M,\Wedge{2} \mathcal{V}^\partial)$ to \emph{gauge-fix} $\bem_a^0$. Consider $\mathbf{L}_{c^{0\underline{i}}}$. Its characteristic distribution is given by the Hamiltonian vector field $\mathbb{L}_{c^{0\underline{i}}}$, whose flow is readily computed by means of the equations: 
$$\dot{\bem}_a^0 = [c,\bem_a]^0 = c^{0\underline{i}}\bem_a^{\underline{j}} \eta_{\underline{i}\underline{j}};\ \ \ \ [\dot\bom]^{0\underline{i}} = [d_{\bom}c]^{0\underline{i}}.$$
We can then follow $\mathbb{L}_{c^{0\underline{i}}}$ up to time $t=1$, fix $\bem_a^{0}(1)=0$ and plug $c^{0\underline{i}}$ in the remaining equation to solve for $[\bom]^{0\underline{i}}$. We will denote the transformed connection at time $t=1$ by $[\underline{\bom}]^{0\underline{i}}$.

Since we have gauge-fixed $\bem_a^0=0$, Equation \eqref{splitconstraints} becomes
\begin{equation}\label{splitconstraints2}
\begin{cases}
\tem_{\langle a}^{\langle\underline{i}} \partial_b \tem_{c\rangle}^{\underline{j}\rangle} + \tem_{\langle a}^{\langle\underline{i}}\bom_b^{\underline{j}\rangle\underline{k}}\eta_{\underline{k}\underline{\ell}}\tem^{\underline{\ell}}_{c\rangle} =0\\
\tem^{\underline{j}}_{\langle a}\underline{\bom}_b^{0\underline{k}}\eta_{\underline{k}\underline{\ell}}\tem_{c\rangle}^{\underline{\ell}} = 0
\end{cases}
\end{equation}
and results in fixing $[\bom]_a^{\underline{k}\underline{l}}$ as a function of the triad $\tem_a^{\underline{i}}$ - we called that solution $\bg_{\tem}$ - plus a condition on the $0$-part of the connection $\underline{\bom}_a^{0\underline{i}}= : \bA_a^{\underline{i}}$, so the we have $\bem\wedge  d_{\bom}\bem = 0 \iff \bom= \bg_{\bem}^{\underline{i}\underline{j}} w_{\underline{i}}\wedge w_{\underline{j}} + \bA^{\underline{i}}w_0\wedge w_{\underline{i}}$ where
\begin{align}\label{gammaconstraint}
&[\bg_{\tem}]_a^{\underline{i}\underline{j}}\eta_{\underline{i}\underline{k}}\eta_{\underline{j}\underline{\ell}} =  \tem_a^{\underline{q}}\left(\binv_{\underline{q}}^c \binv_{\underline{k}}^d \eta_{\underline{p}\underline{\ell}} + \binv_{\underline{\ell}}^c \binv_{\underline{q}}^d \eta_{\underline{p}\underline{k}} - \binv_{\underline{k}}^c \binv_{\underline{\ell}}^d \eta_{\underline{p}\underline{q}} \right)\partial_c \tem_d^{\underline{p}}\\\label{Aconstraint}
&\bA_{\langle b}^{\underline{k}} \tem_{c\rangle}^{\underline{\ell}}\eta_{\underline{\ell}\underline{k}}=0.
\end{align}

If $\iota_{\mathbf{L}}\colon\mathcal{G}\longrightarrow \mathcal{F}^\partial_{PCH}$ denotes the inclusion of the zero-locus of $\mathbf{L}_{c}$, we have that $\alpha^\partial_{\mathcal{G}}\coloneqq\iota_{\mathbf{L}}^*\alpha^\partial_{PCH}$ is a one-form on $\mathcal{G}$, basic w.r.t.\ the reduction along the characteristic foliation $c\colon \mathcal{G}\longrightarrow \underline{\mathcal{G}}$, i.e. $\alpha^\partial_{\mathcal{G}}=c^*\alpha^\partial_{\underline{\mathcal{G}}}$, and $\varpi_{\underline{\mathcal{G}}}^\partial = \delta \alpha^\partial_{\underline{\mathcal{G}}}$ is the symplectic structure on $\underline{\mathcal{G}}$.

The fields on $\mathcal{G}$ are the $\tem_a^i$'s - the components of the triad $\tem$ - and whatever is left of $\bom$, namely $[\underline{\bom}]^{0\underline{i}}=\bA^{\underline{i}}$ satisfying \eqref{Aconstraint}, which can be interpreted as a (global) $\mathfrak{so}(3)$ connection\footnote{Possibly $\mathfrak{so}(2,1)$ if $\eta_{\underline{i}\underline{j}}$ has residual Lorentzian signature.}. Reduction with respect to the characteristic foliation, i.e. the remaining $SO(3)$ action, will then yield the map
\begin{equation}
\pi\colon \mathcal{F}^\partial_{PCH} \longrightarrow \underline{\mathcal{G}}.
\end{equation}

The PCH boundary one-form on $\mathcal{G}$ reads
\begin{align}\label{projform}
\alpha^\partial_{{\mathcal{G}}} = &\iota_{\mathbf{L}}^*\intl_{\partial M}\frac12\mathrm{Tr}\left[ - \bem\wedge\bem\wedge\delta T_\gamma[\bom] + \delta\left(\bem\wedge\bem \wedge T_\gamma[\bom]\right) \right] \\
= & \intl_{\partial M}\tem_a^{\underline{i}}\delta\tem_b^{\underline{j}}[\mathbf{A}]_c^{\underline{k}}\epsilon_{\underline{i}\underline{j }\underline{k}}\epsilon^{abc} + \frac{1}{\gamma} \tem^{\underline{i}}_a \delta\tem ^{\underline{j}}_b  [\mathbf{\Gamma}_{\tem}]_c^{\underline{k}\underline{l}}\eta_{\underline{i}\underline{k}}\eta_{\underline{j}\underline{l}}\epsilon^{abc}.\end{align}

With a long but straightforward computation one can show that, from the explicit expression of $\bg_{\tem}$ in Equation \eqref{gammaconstraint}, we have (round brackets mean symmetrisation of enclosed indices)
\begin{multline}
\epsilon^{afg} [\bg_{\tem}]_a^{\underline{i}\underline{j}}\eta_{\underline{i}\underline{k}}\eta_{\underline{j}\underline{\ell}} \tem_f^{\underline{k}}\delta\tem_g^{\underline{\ell}} \\
=\left( \partial_a\tem_d^{\underline{p}}\eta_{\underline{p}\underline{\ell}}\binv_{\underline{k}}^d + \partial_d\tem_a^{\underline{p}}\eta_{\underline{p}\underline{k}}\binv_{\underline{\ell}}^d - \partial_c\tem_d^{\underline{p}}\eta_{\underline{p}\underline{q}}\tem_a^{\underline{q}}\binv_{\underline{k}}^c\binv_{\underline{\ell}}^d\right)\tem_f^{\langle\underline{k}}\delta \tem_g^{\underline{\ell}\rangle}\epsilon^{afg}\\
=\partial_{a}\tem_{f}^{\underline{p}}\eta_{\underline{p}\underline{\ell}}\delta\tem_g^{\underline{\ell}}\epsilon^{afg} - \partial_{(a} g_{f)d}\binv_{\underline{k}}^d\delta \tem_g^{\underline{k}}\epsilon^{afg} + \partial_{(a}\tem_{f)}^{\underline{\ell}}\eta_{\underline{\ell}\underline{k}}\delta\tem^{\underline{k}}_g\epsilon^{afg} + \\
+ \partial_d g_{af} \binv^d_{\underline{\ell}}\delta\tem^{\underline{\ell}}_g\epsilon^{afg} - \tem_a^{\underline{p}}\eta_{\underline{p}\underline{k}}\partial_d\tem^{\underline{k}}_f\binv^d_{\underline{\ell}}\delta\tem^{\underline{\ell}}_g\epsilon^{afg}\\
= \partial_{a}\tem_{f}^{\underline{p}}\eta_{\underline{p}\underline{\ell}}\delta\tem_g^{\underline{\ell}}\epsilon^{afg} = \frac12 \partial_a\left(\delta\tem_f^{\underline{i}}\eta_{\underline{i}\underline{j}}\tem_g^{\underline{j}}\epsilon^{afg}\right),
\end{multline}
where we used $g_{ab}=\tem_a^{\underline{i}}\eta_{\underline{i}\underline{j}}\tem_b^{\underline{j}}$, so that, up to a $\delta$-exact term,
\begin{equation}
\alpha^\partial_{\mathcal{G}}=\intl_{\partial M}\tem_a^{\underline{i}}\delta\tem_b^{\underline{j}}[\mathbf{A}]_c^{\underline{k}}\epsilon_{\underline{i}\underline{j }\underline{k}}\epsilon^{abc} + \frac{1}{2\gamma} \partial_a\left(\delta\tem_f^{\underline{i}}\eta_{\underline{i}\underline{j}}\tem_g^{\underline{j}}\epsilon^{afg}\right) 
\end{equation}
does not depend on $\gamma$ if the boundary has no boundary of its own (a corner). 
Consider now the definition of the symmetric tensor:
\begin{equation}\label{defK}
[K]_{ab} \coloneqq \tem_{(a}^{\underline{i}} [\mathbf{A}]_{b)}^{\underline{j}}\eta_{\underline{i}\underline{j}} \iff [\mathbf{A}]_{b}^{\underline{j}} = \eta^{\underline{j}\underline{k}}\binv^a_{\underline{k}} [K]_{ab}.
\end{equation}
Notice that Equation \eqref{Aconstraint} is equivalent to $K_{\langle a b\rangle} =0$, and that the definition of $K$ descends to $\underline{\mathcal{G}}$, for the residual action on fields reads $\bem\mapsto [c,\bem]$ and $\underline{\bom}^{0\underline{i}} \mapsto [\bom,c]^{0\underline{i}}$, which then leaves the contraction $\bem^{\underline{i}}\mathbf{A}^{\underline{j}} \eta_{\underline{i}\underline{j}}$ invariant. Furthermore, set $g^\partial = \bem^*\eta \equiv \varphi (\bem)$, which also descends to $\underline{\mathcal{G}}$, and
\begin{equation}\label{Pidef}
\Pi = \frac{\sqrt{\mathsf{g}^\partial}^\partial}{2}\left(K - g^\partial\mathrm{Tr}_{g^\partial}[K]\right) \end{equation}
with $\mathsf{g}^\partial=|\mathrm{det}(g^{\partial})|$. We claim that
\begin{equation} 
\alpha^\partial_{\underline{\mathcal{G}}}= \varphi^* \intl_{\partial M} \Pi_{ab} \delta [g^\partial]^{ab}   = \pi^*\alpha^\partial_{EH},
\end{equation}
and $\varphi\colon \underline{\mathcal{G}} \longrightarrow \mathcal{F}^\partial_{EH}$, defining $(g^\partial,\Pi)$ in terms of $(\bem,\mathbf{A})$ is a symplectomorphism. As a matter of fact, imposing 
\begin{equation}
2\intl_{\partial M} \delta \tem_a^{\underline{i}}\tem_b^{\underline{j}}\bA_c^{\underline{k}}\epsilon_{\underline{i}\underline{j}\underline{k}}\epsilon^{abc} = \varphi^*\alpha^\partial_{EH}= \varphi^*\intl_{\partial M} -\Pi^{ab}\delta g^\partial_{ab} = -2\intl_{\partial M} \Pi^{ab}\delta \tem_a^{\underline{i}}\eta_{\underline{i}\underline{j}}\tem_b^{\underline{j}}
\end{equation}
(recall that $\delta [g^\partial]^{ab}= - [g^\partial]^{ac}\delta g^\partial_{cd}[g^\partial]^{db}$), we have that 
$$\Pi^{ab}\tem_b^{\underline{\ell}} = -\eta^{\underline{\ell}\underline{i}}\epsilon_{\underline{i}\underline{j}\underline{k}}\tem_b^{\underline{j}}\bA_c^{\underline{k}}\epsilon^{abc}.$$
Now, with a simple calculation:
\begin{multline}
\Pi_{cd} = -g^\partial_{a(c} \Pi^{ab}\tem_b^{\underline{\ell}}\eta_{\underline{\ell}\underline{m}}\tem^{\underline{m}}_{d)} = -g_{a(c}\epsilon^{agf}\tem_f^{\underline{j}}\epsilon_{\underline{j}\underline{i}\underline{k}}A^{\underline{k}}\tem^{\underline{i}}_{d)}=-g_{a(c}|\tem| \binv^a_{\langle\underline{i}}\binv^g_{\underline{k\rangle}}\bA^k_g \tem_{d)}^{\underline{i}}\\
=-\frac{|\tem|}{2}\left( g_{cd} \binv^g_{\underline{k}}\bA^{\underline{k}}_g - \tem_{(c}^{\underline{r}}\eta_{\underline{r}\underline{k}}\bA^{\underline{k}}_{d)}\right)=\frac{\sqrt{\mathsf{g}^\partial}}{2}\left(K_{cd} - g^\partial_{cd}\mathrm{Tr}_{g^\partial}K\right),
\end{multline}
since we have that $\mathrm{Tr}_{g^\partial}[K] = g^{cd}K_{cd}=\binv^g_{\underline{k}}\bA^{\underline{k}}_g$, and we used 
\begin{equation}\label{detformula}
|\bem| = \tem_c^{\underline{l}}\tem_a^{\underline{i}}\tem_b^{\underline{j}}\epsilon_{\underline{i}\underline{j}\underline{l}}\epsilon^{abc} \iff |\tem|\binv^b_{\underline{j}}  = \tem_c^{\underline{l}}\tem_a^{\underline{i}}\epsilon_{\underline{i}\underline{j}\underline{l}}\epsilon^{abc}.
\end{equation}

To show that the vanishing locus of $\mathbf{J}_\lambda=\intl_{\partial M} \mathrm{Tr} \left[\lambda \bem T_\gamma[F_{\bom}]\right]$ in $\underline{\mathcal{G}}$ is isomorphic to that of the Einstein--Hilbert constraints, we start again from the decomposition $\bom_\gamma = [\bg_{\bem}]^{\underline{i}\underline{j}} w_i\wedge w_j + [\bA]^{\underline{i}}w_0\wedge w_i$ (valid on $\mathcal{G}$) where $\bg_{\bem}$ represents the spin connection compatible with the residual triad $\bem$ (after fixing $\bem_a^0=0$) and $\bA$ satisfies Equation \eqref{Aconstraint}. 

We can rewrite the constraint $\mathbf{J}_\lambda$ as
\begin{equation}
\mathbf{J}_\lambda = \intl_{\partial M} \lambda^i \bem^j F_{\bom}^{kl}\epsilon_{ijkl} + \frac{1}{\gamma}\lambda^i\bem^j F_{\bom}^{kl}\eta_{ik}\eta_{jl} = \mathbf{J}_\lambda^{\infty} + \frac{1}{\gamma}\intl_{\partial M}\eta(\lambda, [F_{\bom},\bem]).
\end{equation}
We claim that $\mathbf{J}_\lambda\vert_{\mathcal{G}}=\mathbf{J}^\infty_\lambda\vert_{\mathcal{G}}$. On $\mathcal{G}$, in fact, we can decompose $F_{\bom}=F_{\bg_{\tem}} + d_{\bg_{\tem}}\bA + \frac12 [\bA,\bA]$, and we have
\begin{multline}
\intl_{\partial M}\lambda^i\bem^j F_{\bom}^{kl}\eta_{ik}\eta_{jl} \Big|_{\mathcal{G}} = \intl_{\partial M}\lambda^{\underline{i}}\tem^{\underline{j}}[F_{\bg_{\tem}}]^{\underline{k}\underline{\ell}}\eta_{\underline{i}\underline{k}}\eta_{\underline{j}\underline{\ell}} + \lambda^i \bem^j[d_{\bg_{\tem}}\bA]^{kl}\eta_{ik}\eta_{jl} \\
+ \lambda^i\bem^j\eta_{00}\bA^{\underline{k}}\bA^{\underline{\ell}}\eta_{i\underline{k}}\eta_{j\underline{\ell}}\Big|_{\mathcal{G}}\\
=\intl_{\partial M} \lambda^i [F_{\bg_{\tem}},\tem]^{\underline{k}}\eta_{\underline{i}\underline{k}} + [d_{\bg_{\tem}}\lambda]^i \tem^{\underline{j}}\bA^{kl}\eta_{ik}\eta_{jl} + \eta_{00}\lambda^{\underline{i}}\bA^{\underline{k}}\eta_{\underline{i}\underline{k}} \tem^{\underline{j}}\bA^{\underline{\ell}}\eta_{\underline{j}\underline{\ell}} = 0,
\end{multline}
where we must use the Bianchi identity for $\bg_{\tem}$ on the first term, and Equation \eqref{Aconstraint} to show that every occurrence of $\eta(\tem_{\langle a},\bA_{b\rangle})$ vanishes. For clarity, we unwrap the second term 
$$[d_{\bg_{\tem}}\lambda]_a^i \tem_b^{\underline{j}}\bA_c^{kl}\eta_{ik}\eta_{jl}\epsilon^{abc} = \eta_{00}[d_{\bg_{\tem}}\lambda]_a^0 \tem_b^{\underline{j}}\bA_c^{\underline{\ell}}\epsilon^{abc}\eta_{\underline{j}\underline{\ell}}=0.$$

Now, the splitting induced by $w_0$ on $\bom$ applies to the curvature form as well, i.e. $F_{\bom}=[F_{\bom}]^{\underline{i}\underline{j}} w_i\wedge w_j + [F_{\bom}]^{0\underline{i}} w_0\wedge w_i$. It is easy to gather that we can decompose 
$$[F_{\bom}] = \left[[F_{\bg_{\tem}}]^{\underline{i}\underline{j}} + \frac12 \bA^{0\underline{i}}\bA^{0\underline{j}}\eta_{00}\right] u_i\wedge u_j + d_{\bg_{\tem}}\bA^{0\underline{i}}u_0\wedge u_i.$$

Consequently, we can rewrite the constraint $\mathbf{J}^\infty_\lambda$ on $\mathcal{G}$ as
\begin{equation}
\mathbf{J}^\infty_\lambda\Big|_{\mathcal{G}} = \intl_{\partial M} \lambda^0 \tem^{\underline{k}}\left[[F_{\bg_{\tem}}]^{\underline{i}\underline{j}} + \frac12\bA^{0\underline{i}}\bA^{0\underline{j}}\eta_{00}\right] \epsilon_{0\underline{k}\underline{i}\underline{j}} + \lambda^{\underline{k}}\tem^{\underline{j}}d_{\bg_{\tem}}\bA^{0\underline{i}} \epsilon_{\underline{k}\underline{j}0\underline{i}},
\end{equation}
which also splits in $\mathbf{J}^\infty_{\lambda^0}$ and $\mathbf{J}^\infty_{\lambda^{\underline{i}}}$. Observe that $\eta_{00}$ is the signature of the chosen internal direction: if $u_0$ is a time-like vector we will have $\eta_{00}=-1$. Furthermore, it is well known that the combination $\bem^{\underline{k}}[F_{\bg_{\bem}}]^{\underline{i}\underline{j}}\epsilon_{\underline{k}\underline{i}\underline{j}} = \sqrt{\mathsf{g}^\partial}R[\bem^*\bg_{\bem}]=: \sqrt{\mathsf{g}^\partial}R^\partial$, the Ricci scalar density associated with the Levi--Civita connection of the metric $g^\partial=\bem^*\eta$. Using this, in combination with \eqref{detformula} and from the definition of the symmetric tensor $K$ we have
\begin{multline*}
\mathbf{J}^\infty_{\lambda^0}\Big|_{\mathcal{G}} = \intl_{\partial M} \lambda^0\left[ \sqrt{\mathsf{g}^\partial}R^\partial + 2\eta_{00}|\tem| \binv^b_{\langle\underline{i}}\binv^c_{\underline{i}\rangle} \bA^{\underline{i}}_b\bA^{\underline{j}}_c \right] = \\
= \intl_{\partial M} \lambda^0\left[ \sqrt{\mathsf{g}^\partial}R^\partial + \eta_{00}|\tem| \left( \binv^b_{\underline{i}} \bA^{\underline{i}}_b \binv^c_{\underline{i}}\bA^{\underline{j}}_c + \binv^c_{\underline{i}} \bA^{\underline{i}}_b \binv^b_{\underline{j}} \bA^{\underline{j}}_c \right) \right] = \\
= \intl_{\partial M} \lambda^0\left[ \sqrt{\mathsf{g}^\partial} R^\partial + \eta_{00} |\tem|\left( [g^\partial]^{ab}K_{ab} [g^\partial]^{cd}K_{cd} - [g^\partial]^{ab}K_{bd}[g^\partial]^{dc}K_{fc}\right) \right]\\
=\intl_{\partial M} \lambda^0 \left[\sqrt{\mathsf{g}^\partial} R^\partial + \eta_{00} |\tem|\left(\mathrm{Tr}_{g^\partial}[K]^2 - \mathrm{Tr}_{g^\partial}[K^2]\right)\right],
\end{multline*}
and after using the definition\footnote{We define the expression $\mathrm{Tr}_{g^\partial}[\Pi^k]$ to be the trace of $\Pi^k(x)$ seen as an endomorphism of the tangent space at $x$, e.g. $\mathrm{Tr}_{g^\partial}[\Pi^2]=[g^\partial]^{ab}\Pi_{bc}[g^\partial]^{cd}\Pi_{da} = \Pi^a_b\Pi^b_a$.} of $\Pi$ given in \eqref{Pidef}, we can easily gather that $\mathbf{J}^\infty_{\lambda^0}$ reduces to the \emph{Hamiltonian constraint} in Einstein--Hilbert theory \cite{CS1}, for a space--time boundary whose normal vector has signature $\eta_{00}$:
\begin{equation}
\mathbf{J}^\infty_{\lambda^0}\Big|_{\mathcal{G}} =\varphi^* \intl_{\partial M} \lambda^0 \left[ \sqrt{\mathsf{g}^\partial}R^\partial - \eta_{00} \frac{1}{\sqrt{g^\partial}}\left(\mathrm{Tr}_{g^\partial}[\Pi^2] - \frac{1}{d-1}\mathrm{Tr}_{g^\partial}[\Pi]^2 \right)\right].
\end{equation}

With a similar computation we can show that
\begin{multline}
\mathbf{J}^\infty_{\lambda^{\underline{k}}}\Big|_{\mathcal{G}} = \intl_{\partial M} \lambda^{\underline{k}} [d_{\bg_{\tem}}]_b(\tem^{\underline{j}}_c \bA_a^{\underline{i}}\epsilon^{bca}) \epsilon_{\underline{k}\underline{j}0\underline{i}} = \intl_{\partial M} \lambda^{\underline{k}} [d_{\bg_{\tem}}]_b \left(|\tem| \binv^a_{\langle\underline{i}}\binv^b_{\underline{j}\rangle}\bA_a^{\underline{i}}\right) =\\
= \frac12\intl_{\partial M} \xi^{{f}} \tem_f^{\underline{k}} [d_{\bg_{\tem}}]_b \left(|\tem|  \binv^c_{\underline{j}}\bA^{\underline{j}}_c \binv^b_{\underline{k}} - |\tem|\binv^c_{\underline{k}} \bA_c^{\underline{j}}\binv_{\underline{j}}^b\right) = \\
= \frac12\intl_{\partial M} \xi^{{f}} [d_{\bg_{\tem}}]_b \left(\sqrt{\mathsf{g}^\partial}\left(\delta_f^b \mathrm{Tr}_{g^\partial}K - [g^\partial]^{bc}K_{cf}\right) \right) = \intl_{\partial M} \xi^f [d_{\bg_{\bem}}]_b \left([g^\partial]^{bc}\Pi_{cf}\right),
\end{multline}
which coincides with the \emph{momentum constraint} of General Relativity in the Einstein--Hilbert formalism.
\end{proof}

\begin{remark}
An analysis similar to Theorem \ref{PCHtoEH} is given in \cite[Chapter 4]{Thiemann}, where the equivalence between the triadic formulation of the Hamiltonian theory of GR in the Ashtekar variables and the standard Einstein--Hilbert version is established. We have shown here how the canonical boundary theory, as induced by 4-dimensional Palatini--Cartan--Holst theory, is related to that of Einstein--Hilbert theory by means of Marsden--Weinstein reduction for the action of the Lie algebra $\mathfrak{so}(3,1)$ on the space of boundary fields. Splitting the $\mathfrak{so}(3,1)$ (boundary) connection into two $\mathfrak{so}(3)$ pieces, which in \cite{Thiemann} are interpreted as two separate entities, represents here the natural byproduct of considering a 3-dimensional subspace of the inner space $V$ as target of tetrads restricted to the boundary. Observe that the 4-dimensional interpretation presented in \cite[Page 134]{Thiemann} agrees with our construction, and the \emph{Sen--Ashtekar--Immirzi--Barbero} connection (denoted \emph{ibidem} by ${}^{(\beta)}\!A$) coincides with $T_\gamma[\bom]$ under $\beta=\gamma^{-1}$ and rescaling $K$ by $\gamma$.  In our analysis, the Barbero-Immirzi parameter governs the morphism $T_\gamma\colon\Wedge{2} V\longrightarrow \Wedge{2}V$ (cf. Lemma \ref{Lem:twistpairing}), and it disappears only on the critical locus $\mathcal{G}$. The Ashtekar formulation is recovered in our language by just reducing one-half of the constraints, namely $\mathbf{L}_{c^{0\underline{i}}}$, and at that stage the constraint $\mathbf{J}_\lambda$ will still depend on $\gamma$. For this reason, $\gamma$ might still play a significant role in the BFV analysis (which provides a cohomological resolution of the symplectic reduction), and on the induced structure on the corners (boundaries of $\partial M$); we will return on this in a forthcoming paper.
\end{remark}

\begin{remark}
Observe that the result holds so far only for $g^\partial$ nondegenerate. We showed how the reduced phase space of PCH is related to that of EH. In fact, on $\mathcal{F}^\partial_{PCH}$ we may consider the map $\bem \mapsto g^\partial \coloneqq \eta(\bem,\bem)$, which is $\mathfrak{so}(3,1)$ gauge invariant, the equations $\bem\wedge d_{\bom} \bem =0$ are used to fix six components of $\omega$, and the remaining 6 can be related to the extrinsic curvature in EH (the field $K$ in Theorem \ref{PCHtoEH}). The constraints $\bem \wedge T_\gamma[F_{\bom}]=0$ then reduce to the energy and momentum constraints in EH. Note that, since our phase space has 12 local degrees of freedom and we have 10 local constraints, the reduced phase space has 2 local degrees of freedom to be interpreted as the two polarizations of the graviton. As we mentioned in Remark \ref{remarkkerneldegrees}, when $g^\partial$ is degenerate we expect two more constraints to fix the extra degrees of freedom that do not allow us to project the constraints to $\mathcal{F}^{\partial}_{\partial M}$. Assuming that the resulting constraint submanifold $C_{\partial M}$ is still coisotropic (cf. section \ref{Sect:Constraintstheory}), we would have that in case of a light-like boundary the reduced phase space has no local degrees of freedom. This is an interesting exploration that will be considered elsewhere.
\end{remark}

\subsection{Half-Shell PCH}\label{clhalfshell}
In this section we will consider the explicit localisation to the \emph{half-shell} submanifold $d_\omega e=0$ \emph{via} a Lagrange multiplier $t\in \Omega^2(M, \iota^*\mathcal{V}^*)$, where we identify the fibres of $\mathcal{V}^*$ with $\bigwedge^3 V$. The resulting theory is that of Definition \ref{Def:HSPCH}, and we recall the action functional:
\begin{equation}\label{PCHconstraintaction}
S_{HS}=\int\limits_{ M} \frac12\hat{T}_\gamma\left[e\wedge e\wedge F_\omega\right] + \mathrm{Tr}\left[t\wedge d_\omega e\right] + \frac{\Lambda}{4}\mathrm{Tr}\left[e^4\right]
\end{equation}

\begin{theorem}\label{Theo:halfshellcl}
Half-Shell Palatini--Cartan--Holst theory is classically equivalent to Palatini--Cartan--Holst theory, and it has the symplectic space of boundary fields 
\begin{equation}
\mathcal{F}^{\partial}_{HS} = T^*\Omega_{\text{nd}}^1(\partial M,\mathcal{V}^\partial)
\end{equation}
The surjective submersion $\pi_{HS}\colon \mathcal{F}_{HS}\longrightarrow \mathcal{F}^{\partial}_{HS}$ has the explicit expression
\begin{equation}\label{projHS}
\pi_M\colon\begin{cases}
\bt=t + T_\gamma[\underline\omega - \omega]\wedge e\\
\bem=e
\end{cases}
\end{equation}
where $\underline\omega$ is the bulk Levi--Civita $e$-compatible connection restricted to the boundary (i.e. $\underline\omega=\omega(e)|_{\partial M}$), and the (exact) symplectic form in this local chart reads
\begin{equation}{
\varpi^\partial_{HS}=\int\limits_{\partial M} \mathrm{Tr} \left[\delta \bt \wedge \delta \bem\right]}
\end{equation}

Moreover, there exists a local symplectomorphism $\varphi\colon \mathcal{F}^{\partial}_{HS}\longrightarrow \mathcal{F}^{\partial}_{PCH}$ to the space of boundary fields of Palatini--Cartan--Holst theory of Theorem \ref{Theo:PCHClassical}.
\end{theorem}

\begin{proof}
First of all let us analyse the Euler--Lagrange equations for the action functional $S_{HS}$. They read:
\begin{subequations}\label{ELHSPAL}\begin{align}
\delta\omega&: T_\gamma[d_\omega e \wedge e] - t\stackrel{\circ}{\wedge} e = 0\\
\delta e&: e\wedge T_\gamma\left[F_{\omega}\right] + d_\omega t + \Lambda e^3=0\\\label{HSC}
\delta t&: d_\omega e=0
\end{align}\end{subequations}
where $t\stackrel{\circ}{\wedge}e$ stands for $\frac{\delta}{\delta \omega} (t\wedge [\omega, e])=t_{\mu\nu}^{ijk}e_\rho^{m}\epsilon_{ijkl}\epsilon^{\mu\nu\rho\sigma}$ for all $\sigma$ space-time indices and $l,m$ internal indices. Enforcing the half shell constraint $d_\omega e=0$, which implies that $\omega$ is the Levi--Civita connection, represented in the tetrad formalism by the special connection $\omega_e$, we obtain $t=0$ from $t\stackrel{\circ}{\wedge}e=0$, and the Einstein equation in the tetrad formalism
\begin{equation}
e\wedge T_\gamma[F_{\omega_e}] +  \Lambda e^3=0
\end{equation}

From the computations in Theorem \ref{Theo:PCHClassical}, we get that the pre-boundary two-form $\widetilde{\varpi}_{HS}$ reads
\begin{equation}
\widetilde{\varpi}_{HS}=\int\limits_{\partial M} \hat{T}_\gamma\left[\delta e\wedge e \wedge \delta \omega\right] - \mathrm{Tr}\left[\delta t\wedge \delta e\right]
\end{equation}
and the kernel of this two-form is easily found to be:
\begin{subequations}\label{kerHS}\begin{align}
(X_t)&= (X_{T_\gamma\left[\omega\right]})\wedge e\\
(X_e)&=0
\end{align}\end{subequations}
This means that $\omega$ can be fixed using the vertical vector field
\begin{equation}
\mathbb{\Omega}=\X{T_\gamma\left[\omega\right]}\pard{}{T_\gamma\left[\omega\right]} + \X{T_\gamma\left[\omega\right]}\wedge e\pard{}{t}
\end{equation}
and $t$ is modified accordingly. Flowing along $\mathbb{\Omega}$ we can set $\omega$ to be a background connection $\underline\omega$, which we may eventually choose to be the restriction to the boundary of the solution $\omega_e$ of the half-shell constraint \eqref{HSC}, and this fixes $\X{T_\gamma\left[\omega\right]}=T_\gamma\left[\underline\omega - \omega_0\right]$. Then, by solving the straightforward differential equation $\dot{t}= T_\gamma[\underline\omega - \omega_0]\wedge e_0$:
\begin{equation}
t(s)=t_0 + T_\gamma[\underline\omega - \omega_0]\wedge e_0\ s
\end{equation}
we set $t(1)=t_0 + T_\gamma[\underline\omega - \omega_0]\wedge e_0$. This gives us the explicit map to the space of boundary fields $\mathcal{F}_{HS}^{\partial}$
\begin{equation}
\pi_M\colon\begin{cases}
\bt = t + T_\gamma [\underline\omega - \omega]\wedge e\\
\bem = e
\end{cases}
\end{equation}

Notice, however, that the pre-boundary one-form is not horizontal with respect to the kernel foliation defined by equations \eqref{kerHS}, as the generator $\mathbb{\Omega}=\pard{}{\omega}$ does not lie in the kernel of $\widetilde{\alpha}$. We can nevertheless modify $\widetilde{\alpha}$ by adding the exact term $\frac12\int_{M} d\,\hat{T}_\gamma\left[e\wedge e\wedge \omega\right] + \int_M\mathrm{Tr}\left[d(e\wedge t)\right]$ to the action \eqref{PCHconstraintaction}, yielding
\begin{equation}
\overline{\alpha}_{HS} = \widetilde{\alpha}_{HS} + \intl_{\partial M}\mathrm{Tr}\left[\frac12\delta(e\wedge e\wedge T_\gamma[\omega]) + \delta(e\wedge t)\right]
\end{equation}
and it is easy to gather that the following one form on the space of boundary fields
\begin{equation}
\alpha^\partial_{HS} = \int\limits_{\partial M}\mathrm{Tr} \left[\bt\wedge\delta \bem\right]
\end{equation}
will be such that
\begin{equation}
\overline{\alpha}_{HS}=\pi_M^*\alpha^\partial_{HS}
\end{equation}
and $\varpi_{HS}^\partial=\delta \alpha^\partial_{HS}$.

Finally, observe that the map   
$$\varphi\colon\Omega^2(\partial M, \Wedge{3}\mathcal{V}^\partial) \longmapsto \mathrm{Im}(\mathsf{W}_e^{(2)})\simeq \qsp{\mathcal{A}_{\iota^*P}}{\mathrm{Ker}(\mathsf{W}^{(2)}_e)}\equiv \mathcal{A}_{\iota^*P}^{red}$$ 
(which requires the choice of a reference connection) such that $\bt \mapsto T_\gamma[\bom]\wedge \bem$, extends to a local symplectomorphism between $\mathcal{F}_{HS}^{\partial}$ and $\mathcal{F}_{PCH}^{\partial}$.
\end{proof}

\begin{proposition}\label{HSnonlagrangian}
The projection to the space of classical boundary fields of the Euler--Lagrange equation for the action \eqref{PCHconstraintaction} is isotropic but not Lagrangian.
\end{proposition}

\begin{proof}
Consider the Euler--Lagrange equations for the Half-Shell-constrained Palatini--Cartan action as given in \eqref{ELHSPAL}. Their projection to the space of pre-boundary fields is given by:
\begin{equation}
\widetilde{\pi}(EL_{HS})\coloneqq \left\{(e, \omega, t)\in \widetilde{\mathcal{F}}_{HS}^{\partial}\ \big|\ 
\omega=\underline{\omega};\ t=0;\ e\wedge T_\gamma[F_{\underline\omega}] + \Lambda e^3=0\right\}
\end{equation}
Now, taking into account the projection to the space of boundary fields \eqref{projHS} we can gather that the projected critical locus reads
\begin{equation}\label{HSEL}
L_{HS}^\partial\coloneqq\pi_{HS}(EL_{HS})\coloneqq \left\{(\mathbf{t},\bem)\in\mathcal{F}_{HS}^{\partial}\ \big|\ \bt=0;\ \bem\wedge T_\gamma[F_{\underline\omega}] + \Lambda \bem^3=0 \right\}
\end{equation}
since $\bt=t + T_\gamma\left[\underline\omega - \omega\right]\wedge e$. It is easy to check that $L_{HS}^\partial$ is isotropic, as $\bt=0$ implies $\omega^\partial\big|_{EL^\partial}=0$.

Actually, $\bt=0$ defines a Lagrangian submanifold, which is then spoiled by equation $\bem\wedge T_\gamma[F_{\underline\omega}] + \Lambda \bem^3=0$. A way to see this is by explicitly checking that their Poisson bracket is not proportional to the constraints, and thus $L_{HS}^\partial$ fails to be a Lagrangian submanifold.
\end{proof}

\begin{remark}
From this we can see that the classical equivalence of theories for (closed) manifolds without boundary is a notion that requires refinement. Choosing the Palatini--Cartan--Holst action over its half-shell version might be justified by \emph{simplicity}, but it stops being so when boundaries are considered. 

Although the spaces of boundary fields in the PCH case and the Half-Shell version are locally symplectomorphic, the images of the critical loci on the boundary are not mapped into one another. As a matter of fact, the pullback of the locus \eqref{HSEL} yields
$$
\varphi^*L^\partial_{HS}\coloneqq\left\{(\bom,\bem)\in\mathcal{F}_{PCH}^{\partial}\ \big|\ \bom\wedge \bem = 0;\ 
\bem\wedge T_\gamma[F_{\underline\omega}] + \Lambda \bem^3=0\right\},
$$
whereas the (Lagrangian) submanifold for the PCH theory should be given by
$$
L_{PCH}^\partial = \{(\bom,\bem)\ \big|\ \bom=\underline\omega;\ \bem\wedge T_\gamma[F_{\underline\omega}] + \Lambda\bem^3=0\}.
$$

However, note that the two theories are indeed equivalent if we require $t$ to vanish on the boundary. In this case, the boundary analysis is clearly the same as in PCH. This vanishing boundary condition for Lagrange multipliers might be a general feature for constraints depending on derivatives of the fields.
\end{remark}


\begin{thebibliography}{99}

\bibitem[Ash]{Asht} A. Ashtekar, {\it New Variables for Classical and Quantum Gravity}, Phys. Rev. Lett. {\bf 57}, 18 (1986).
\bibitem[Bar]{Bar} J. F. Barbero G., {\it Real Ashtekar variables for Lorentzian signature space-times}, Phys. Rev. {\bf D 51}, 5507 (1995).
\bibitem[BF83]{BFV1} I. A. Batalin and E. S. Fradkin, \emph{A generalized canonical formalism and quantization of reducible gauge theories}, Phys. Lett. {\bf B 122}(2), 157-164 (1983).
\bibitem[BV77]{BFV2} I. A. Batalin and G. A. Vilkovisky, \emph{Relativistic S-matrix of dynamical systems with boson and fermion costraints}, Phys. Lett. {\bf B 69}(3), 309-312 (1977).
\bibitem[BV81]{BV81} I. A. Batalin and G. A. Vilkovisky. {\it Gauge algebra and quantization}, Phys. Lett. {\bf B 102}(1), 27-31 (1981).
\bibitem[BH]{BH} M. Blagojevi\'c and F. W. Hehl, {\it  Gauge Theories of Gravitation, a reader with commentaries}, Imperial College Press (2013). 
\bibitem[Car]{Cartan} E. Cartan. {\it Sur une g\'en\'eralisation de la notion de courbure de Riemann et les espaces \'a torsion}, C. R. Acad. Sci. {\bf 174}, 593?595 (1922).\\
E. Cartan, Comptes rendus hebdomadaires des s\'eances de l'Acad\'emie des sciences,  174, 437-439, 593-595, 734-737, 857-860, 1104-1107 (January 1922).
\bibitem[CMR14]{CMR1} A.S. Cattaneo, P. Mn\"ev, N. Reshetikhin, {\it Classical BV theories on manifolds with boundary}, Commun. Math. Phys. 332 (2): 535-603 (2014).
\bibitem[CMR11]{CMRCorfu} A. S. Cattaneo. P. Mnev and N. Reshetikin, Classical and Quantum Lagrangian Field Theories with Boundary, Proceedings of the Corfu Summer Institute 2011 School and Workshops on Elementary Particle Physics and Gravity, Corfu, Greece, 2011, PoS(Corfu2011)044. 
\bibitem[CS15]{CS1} A. S. Cattaneo and M. Schiavina, {\it BV-BFV approach to General Relativity, Einstein-Hilbert action}, J. Math. Phys. {\bf 57}(2) (2015).
\bibitem[CS17]{CS2} A. S. Cattaneo and M. Schiavina, {\it BV-BFV approach to General Relativity: Palatini--Cartan--Holst action}, in preparation.
\bibitem[CS17]{CS3} A. S. Cattaneo and M. Schiavina, {\it On time}, Lett. Math. Phys. {\bf 107}(2), 375-408 (2017).
\bibitem[CSS]{CSS} A. S. Cattaneo, M. Schiavina and I. Selliah, {\it BV equivalence between triadic gravity and BF theory in three dimensions}, I. Lett Math Phys (2018). https://doi.org/10.1007/s11005-018-1060-5.
\bibitem[Dir]{Dirac} P. A. M. Dirac, {\it Generalized Hamiltonian dynamics}, Canad. J. Math. {\bf 2}, 129-148 (1950).
\bibitem[Eins]{Einstein} A. Einstein, Einheitliche Feldtheorie yon Gravitation und Elektrizit\"at, Sitzungsber. Pruess. Akad. Wiss., 414 (1925). 
\bibitem[FFR]{FFR} M. Ferraris, M. Francaviglia and C. Reina, {\it Variational Formulation of General Relativity from 1915 to 1925 ``Palatini's Method'' Discovered by Einstein in 1925}, General Relativity and Gravitation {\bf14} (3), (1982).

\bibitem[FP]{FloPer} R. Floreanini and R. Percacci, {\it Palatini formalism and new canonical variables for GL(4)-invariant gravity}, Class. Quant. Grav. {\bf 7} 1805 (1990).
\bibitem[Ham]{NM} R. S. Hamilton, {\it The inverse function theorem of Nash and Moser}, Bull. Amer. Math. Soc. {\bf 7} (1), (1982).
\bibitem[HMS]{HMS} R. Hojman, C. Mukku, and W. A. Sayed, {\it Parity violation in metric-torsion theories of gravitation}, Phys. Rev. {\bf D 22}, 1915 (1980)
\bibitem[Hol]{Holst} S. Holst, {\it Barbero's Hamiltonian derived from a generalized Hilbert-Palatini action}, Phys. Rev. {\bf D 53}, 5966 (1996).
\bibitem[Imm]{Imm} G. Immirzi, {\it Real and complex connections for canonical gravity}, Class. Quant. Grav. {\bf 14}, L177 (1997).
\bibitem[Kib]{Kib} T. W. B. Kibble, {\it Lorentz Invariance and the Gravitational Field}, J. Math. Phys. {\bf 2}, 212 (1961).
\bibitem[KT]{KT} J. Kijowski and W. M. Tulczyjew, \emph{A Symplectic Framework for Field Theories}, Lecture notes in Physics {\bf 107}, Springer-Verlag Berlin Heidelberg (1979).

\bibitem[MW]{MW} J. Marsden and A. Weinstein, Reduction of symplectic manifolds with symmetry, Rep. Math. Phys. 5 (1974), 121?130.
\bibitem[Pal]{Palatini} A. Palatini, {\it Deduzione invariantiva delle equazioni gravitazionali dal principio di Hamilton}, Rend. Circ. Mat. Palermo {\bf 43}, 203 (1919).

[English translation by R.Hojman and C. Mukku in P.G. Bergmann and V. De Sabbata (eds.) Cosmology and Gravitation, Plenum Press, New York (1980)].

\bibitem[PR]{Per} A. Perez and D. J. Rezende, {\it Four-dimensional Lorentzian Holst action with topological terms}, Phys. Rev. {\bf D 79}, 064026 (2009).
\bibitem[RT]{Rovth} C. Rovelli and T. Thiemann, {\it Immirzi parameter in quantum general relativity}, Phys. Rev. {\bf D 57}, 1009 (1998).
\bibitem[Sch09]{Schaetz09} F. Schaetz, {\it BFV-complex and higher homotopy structures}, Comm. Math. Phys. {\bf 286}(2), 399-443 (2009).
\bibitem[Sch10]{Schaetz10} F. Schaetz, {\it Invariance of the BFV complex}, Pac. J. Math. {\bf 248} (2), (2010).
\bibitem[Schi]{THESIS} M. Schiavina, {\it BV-BFV  approach to General Relativity}, PhD Thesis, University of Z\"urich (2016).
\bibitem[Sci]{Sci} D. Sciama, {\it The Physical Structure of General Relativity}, Rev. Mod. Phys. {\bf 36}, 463 (1964).
\bibitem[Thi]{Thiemann} T. Thiemann, {\it Modern Canonical Quantum General Relativity} (Cambridge Monographs on Mathematical Physics), Cambridge University Press, 2008. 
\bibitem[Wis]{Wise} D. Wise, {\it Symmetric Space Cartan Connections and Gravity in Three and Four Dimensions}, SIGMA {\bf 5} (2009), 080.

\end{thebibliography}
\end{document}